\documentclass[letterpaper,twocolumn,10pt]{article}
\usepackage{usenix}

\usepackage{amsmath,amssymb,amsthm,mathtools}
\usepackage{graphicx,booktabs}
\usepackage{tikz}
\usetikzlibrary{arrows.meta,positioning,calc,decorations.pathreplacing}
\usepackage{enumitem}
\usepackage{url}

\newif\ifarxiv
\arxivtrue
\setlist[itemize]{noitemsep}

\allowdisplaybreaks
\newtheorem{theorem}{Theorem}[section]
\newtheorem{lemma}[theorem]{Lemma}
\newtheorem{proposition}[theorem]{Proposition}
\newtheorem{corollary}[theorem]{Corollary}
\newtheorem{conjecture}[theorem]{Conjecture}
\theoremstyle{definition}
\newtheorem{definition}[theorem]{Definition}
\theoremstyle{remark}
\newtheorem{remark}[theorem]{Remark}

\newcommand{\mem}{\mathrm{mem}}
\newcommand{\Ext}{\mathrm{Ext}}
\newcommand{\Hloc}{\mathcal{H}_{\mathrm{loc}}}
\newcommand{\D}{\mathcal{D}}
\newcommand{\E}{\mathbb{E}}

\newcommand{\eps}{\varepsilon}
\newcommand{\kap}{\kappa}
\newcommand{\eqdef}{\vcentcolon=}

\definecolor{hsBlue}{HTML}{4053A3}
\definecolor{hsRust}{HTML}{B5522D}
\definecolor{hsTeal}{HTML}{0F7A6D}
\definecolor{hsViolet}{HTML}{7A4FA3}
\definecolor{hsGray}{HTML}{7C8388}
\definecolor{hsRed}{HTML}{A01F2D}

\begin{document}

\date{}

\title{\Large \bf Memorization Is Not Extraction:\\
Tight Differential-Privacy Bounds and Audit Blind Spots}

\ifarxiv
\author{
{\rm Xujun Che}\\
University of North Carolina at Charlotte\\
xche@charlotte.edu
\and
{\rm Depeng Xu}\\
University of North Carolina at Charlotte\\
dxu7@charlotte.edu
\and
{\rm Shuhan Yuan}\\
Utah State University\\
shuhan.yuan@usu.edu
} 
\fi

\maketitle

\begin{abstract}
Memorization in large language models is measured through a zoo of definitions whose formal relations are unknown, and differential privacy (DP) is treated as a proxy against all of them at once. We pin down the exact DP constant for the two that carry the practical weight, counterfactual memorization and adaptive extraction, and show that they do not control each other. Under $f$-DP, every adaptive extraction protocol with list budget $m$ succeeds with probability at most $1-f(\kap)$ for the oblivious baseline $\kap$, and the bound is tight on a dense set of baselines: DP \emph{uniformly} controls extraction exactly up to a threshold in how well the secret can be guessed a priori. Min-entropy certifies that baseline distribution-free, since $H_\infty\ge\eps\log_2 e+\log_2(m/\tau)$ holds extraction below a risk level $\tau\le1/2$ under pure $\eps$-DP for \emph{every} prior, and is exact on uniform priors. On the memorization side, $f$-DP caps the counterfactual memorization of any bounded score at an advantage functional $\eta(f)$, equal to $\tanh(\eps/2)$ under pure DP; for $k\ge2$ duplicated copies the naive $\eps\mapsto k\eps$ bound $\tanh(k\eps/2)$ is unattainable, the exact constant being a closed-form staircase attained by geometric noisy counting. That cap is attained inside the local score class used in practice, and it is there that the two measures separate: one mechanism is memorized yet unextractable, another fully extractable yet exactly invisible to every loss-based score. The two-sided blind spot this opens for loss-based auditing and unlearning verification survives on billion-parameter models: a reserved-trigger release is recovered verbatim from one prompt while the audits practitioners deploy certify it clean.
\end{abstract}

\section{Introduction}\label{sec:intro}

Whether a model has \emph{memorized} a piece of its training data is now a
question with legal and regulatory force: erasure requests under the GDPR
presuppose that removal is verifiable, copyright litigation turns on
whether verbatim reproduction evidences copying
\cite{cooper-grimmelmann,elkin-koren}, and deployment policies increasingly
require quantitative memorization assessments. The empirical record
gives such assessments little to stand on. Training data can be extracted
verbatim from public models \cite{carlini21} and, via divergence-style
prompting, from aligned production systems \cite{nasr23}; measured
memorization grows systematically with model scale, example duplication,
and context length \cite{carlini-quant,tirumala,biderman}, is only
partially mitigated by deduplication \cite{lee-dedup,kandpal}, and is
masked rather than removed by verbatim-output filters \cite{ippolito}.
Models certified as ``unlearned'' recover the supposedly forgotten text
after light fine-tuning on loosely related public data
\cite{hu-relearn,lucki}, and audits based on membership-inference losses
can disagree with extraction-based audits on the same model
(Section~\ref{sec:verif-llm}).

Underlying this confusion is a definitional one. The literature measures
memorization through several coexisting formalizations: canary exposure
\cite{carlini-secret}, eidetic memorization \cite{carlini21},
counterfactual memorization \cite{feldman-zhang,zhang-cf}, discoverable
extraction and its probabilistic refinement
\cite{carlini-quant,hayes-prob}, and distributional memorization, which
asks not about any single record but about corpus-level statistical
structure the model reproduces; recent
surveys present these side by side, with no quantitative reductions or
separations between them \cite{landscape,satvaty}. At the same time, differential
privacy (DP) \cite{dmns,dwork-roth} and low membership-inference (MI) measurements
\cite{shokri,yeom,lira} are widely treated as proxies protecting against
\emph{all} of these notions at once. The proxy logic has solid foundations
for two of the games: the exact DP-to-MI region is classical
\cite{yeom,kov,drs}, and reconstruction is controlled by
reconstruction-robustness bounds \cite{bch22,guo22,hayes-rero,kaissis}.
It has no such foundation for the measures actually used to study
language models, and a recent separation shows it can fail outright: indistinguishability
guarantees neither imply nor are implied by inextractability \cite{lex26}.

This paper provides the missing map (Fig.~\ref{fig:map}). We cast each
notion as an adversary game over a common training pipeline
(Section~\ref{sec:prelim}), prove quantitative bridges from $f$-DP to the
two measures that carry the practical weight, and prove that those two
are themselves incomparable on the score class used in practice. Both
halves are exact rather than asymptotic: the bridges carry constants that
are attained, and the incomparability is witnessed in both directions by
explicit mechanisms. That exactness is what makes the map usable. It
turns the proxy question into a threshold an auditor can evaluate, and it
locates the counterexamples of \cite{lex26} rather than merely
accommodating them: they occupy the low-entropy region where our boundary
says a DP guarantee alone can no longer uniformly control extraction.

\paragraph{Contributions.}
\begin{itemize}[leftmargin=*, itemsep=2pt]
  \item \textbf{A unified game framework} (Section~\ref{sec:prelim}): each
  memorization notion is an adversary game over one pipeline; we isolate
  the \emph{local score class} $\Hloc$ capturing every loss-based
  estimator used in practice \cite{feldman-zhang,zhang-cf,carlini-quant},
  and define quantitative implication and separation between measures
  (Table~\ref{tab:map}).
  \item \textbf{Calibration} (Section~\ref{sec:t1}, Theorem~\ref{thm:t1}):
  $f$-DP bounds counterfactual memorization by the total-variation
  parameter $\eta(f)=\sup_\alpha(1-\alpha-f(\alpha))$ of
  \cite{kaissis-mech,kulynych}, with closed forms
  $\tanh(\eps/2)$, $\delta+(1-\delta)\tanh(\eps/2)$ and $2\Phi(\mu/2)-1$
  under pure, approximate and $\mu$-GDP respectively, and
  tightness attained \emph{within} $\Hloc$. For $k$-duplicated canaries
  the naive $\eps\mapsto k\eps$ bound $\tanh(k\eps/2)$ is
  \emph{unattainable} for $k\ge2$; Lemma~\ref{lem:staircase} solves the
  extremal problem induced by the exact $f$-DP group trade-off
  \cite{drs} in closed form for all $k$, attained by noisy counting.
  \item \textbf{The extraction boundary} (Section~\ref{sec:t2},
  Theorem~\ref{thm:t2}): the reconstruction-robustness bound
  $\Ext\le 1-f(\kap)$ in an adaptive list-budget semantics aligned with
  discoverable extraction \cite{carlini-quant,hayes-prob}, tight on a
  dense set of baselines, with a min-entropy certificate, exact on uniform priors, for when
  DP prevents extraction, and a conditional-prior treatment that charges
  paraphrase-correlated corpora to the baseline.
  Corollary~\ref{cor:audit} turns the converse into a soundness condition
  for canary-based auditing \cite{jagielski,steinke-audit}.
  \item \textbf{Separations} (Section~\ref{sec:t3}, Theorem~\ref{thm:t3}): on
  $\Hloc$, counterfactual memorization and extractability are incomparable
  in both directions, via the uniform camouflage tree (posterior
  uniformity over $k^L$ surviving paths; perfect marginal camouflage) and
  a reserved-trigger side channel whose in-the-wild instances are
  divergence attacks \cite{nasr23} and data-poisoning backdoors
  \cite{badnets,wan}. Corollaries~\ref{cor:blind}--\ref{cor:unlearn} give
  a two-sided blind spot for loss-based auditing and a measure-relativity
  principle for unlearning verification \cite{hu-relearn}.
\end{itemize}

Section~\ref{sec:rw-rero} places these results against the closest prior
line, Section~\ref{sec:axioms} asks which axioms a memorization measure
can satisfy at once, and Section~\ref{sec:verif} tests them at two
scales, ending with
fine-tuned language models on which no audit in the suite we evaluate
catches every leak we plant.

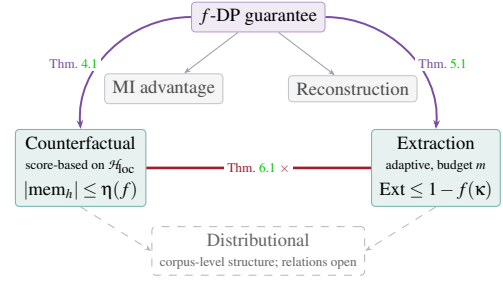
\begin{figure}[t]
\centering
\begin{tikzpicture}[
  font=\scriptsize,
  box/.style={draw, rounded corners=1.5pt, align=center, inner xsep=3pt,
              inner ysep=2.5pt, line width=0.4pt},
  ours/.style={box, fill=hsTeal!10, draw=hsTeal!55},
  guar/.style={box, fill=hsViolet!12, draw=hsViolet!55},
  prior/.style={box, fill=hsGray!8, draw=hsGray!45, text=black!70},
  openbox/.style={box, dashed, fill=white, draw=hsGray!55, text=black!60},
  bridge/.style={-{Stealth[length=3.4pt]}, hsViolet, line width=0.75pt},
  known/.style={-{Stealth[length=2.6pt]}, hsGray!75, line width=0.4pt},
  sep/.style={hsRed, line width=1.0pt},
  tag/.style={font=\tiny, inner sep=1.2pt, fill=white, text=black!75}
]
\node[guar] (fdp) at (0,0) {$f$-DP guarantee};
\node[prior] (mi)   at (-1.25,-0.95) {MI advantage};
\node[prior] (rero) at ( 1.25,-0.95) {Reconstruction};
\node[ours] (mem) at (-2.35,-2.0)
  {Counterfactual\\[-0.5pt] \tiny score-based on $\Hloc$\\[2pt]
   $|\mem_h|\le\eta(f)$};
\node[ours] (ext) at ( 2.35,-2.0)
  {Extraction\\[-0.5pt] \tiny adaptive, budget $m$\\[2pt]
   $\Ext\le 1-f(\kap)$};
\node[openbox] (dist) at (0,-3.1)
  {Distributional\\[-0.5pt] \tiny corpus-level structure; relations open};
\draw[known] (fdp) -- (mi);
\draw[known] (fdp) -- (rero);
\draw[bridge] (fdp.west) to[out=180,in=95, looseness=1.0]
  node[tag, text=hsViolet, pos=0.60, left=-1pt] {Thm.~\ref{thm:t1}} (mem.north);
\draw[bridge] (fdp.east) to[out=0,in=85, looseness=1.0]
  node[tag, text=hsViolet, pos=0.60, right=-1pt] {Thm.~\ref{thm:t2}} (ext.north);
\draw[sep] (mem) -- node[tag, text=hsRed] {Thm.~\ref{thm:t3}\;$\times$} (ext);
\draw[known, dashed, hsGray!55] (mem.south) -- (dist.west);
\draw[known, dashed, hsGray!55] (ext.south) -- (dist.east);
\end{tikzpicture}
\caption{Relations among memorization measures. Solid edges are
proved, dashed ones are open. The violet arrows are the bridges proved
here, labeled with the constant each delivers; the red edge is also a
result, a two-way separation on the local class $\Hloc$, so neither of
those two measures controls the other; thin gray arrows are prior
bridges, and the dashed node is the relation we leave open.}
\label{fig:map}
\end{figure}
\section{Related Work}\label{sec:rw}

\paragraph{Differential privacy and its hypothesis-testing view.}
DP originates with \cite{dmns}; see \cite{dwork-roth}
for the standard treatment and \cite{abadi} for the DP-SGD training
pipeline our results apply to. The hypothesis-testing interpretation
(that DP constrains the type-I/type-II error region of any test
distinguishing neighboring datasets) goes back to \cite{wz10}, was
sharpened into the exact $(\eps,\delta)$ region by \cite{kov}, and is
systematized by the trade-off-function calculus of $f$-DP \cite{drs},
whose formalism (attainment, post-processing, group privacy) our proofs
use throughout; R\'enyi accounting \cite{mironov} is the standard
practical surrogate.

\paragraph{From reconstruction robustness to extraction.}
\label{sec:rw-rero}
The closest antecedent is the reconstruction-robustness (ReRo) framework of Balle,
Cherubin, and Hayes \cite{bch22}, which bounds an informed adversary's
probability of reconstructing a training record under DP as a function
of a prior baseline; \cite{guo22} gives information-theoretic variants,
and \cite{hayes-rero,kaissis} tighten the bounds for DP-SGD and general
$f$-DP. Bounds of this shape are not confined to reconstruction:
\cite{kaissis-mech} identify $\eta(f)=\sup_\alpha(1-\alpha-f(\alpha))$
as the total-variation parameter implied by $f$-DP, and \cite{kulynych}
unify re-identification, attribute inference and reconstruction under
$\text{success}\le1-f(\text{baseline})$, with counterfactual
memorization as a corollary. Theorem~\ref{thm:t1} is that parameter's
exact transfer to counterfactual memorization, including the $k$-copy
refinement where group privacy is provably not tight, and
Theorem~\ref{thm:t2}(i) is the baseline bound in an adaptive
list-budget semantics. Every result in this line is an upper bound, and
none characterizes when the bound is met; Table~\ref{tab:rero} sets out
what follows from closing that gap. The converse also resolves the
apparent tension with \cite{lex26}, which formalizes extraction and
indistinguishability as games and proves them incomparable in the worst
case: its counterexamples live in the low-entropy regime our threshold
identifies as the failure region (Section~\ref{sec:reconcile}).
Theorem~\ref{thm:t3} then adds a second separation, of the
\emph{average-case, score-based} counterfactual measure from
extractability, by posterior-uniformity and side-channel constructions
rather than worst-case gadgets.

\paragraph{Measuring memorization, and auditing it.}
Canary exposure \cite{carlini-secret} and the GPT-2 extraction of
\cite{carlini21} initiated this literature, quantified across scale,
duplication and context length by \cite{carlini-quant} and refined in
\cite{biderman,tirumala}. Probabilistic discoverable extraction
\cite{hayes-prob} replaces the single greedy query by repeated sampling.
Counterfactual memorization was introduced for classification by
\cite{feldman-zhang} and ported to language models by \cite{zhang-cf}.
On the mitigation side, deduplication reduces measured memorization
\cite{lee-dedup,kandpal}, while output filters provide only a false
sense of privacy \cite{ippolito}. Surveys \cite{landscape,satvaty}
present the definitions side by side; the reductions and separations of
Fig.~\ref{fig:map} are new, as are the exact constants behind the first
effect (Corollary~\ref{cor:dup}, Lemma~\ref{lem:staircase}) and the
second (Appendix~\ref{app:d2}). Membership-inference attacks
\cite{shokri,yeom,lira} and their use for auditing DP implementations
\cite{jagielski,nasr21,steinke-audit} are the empirical face of the
DP-to-MI bridge; the sample complexity of auditing is an active
theoretical frontier \cite{haghifam}. Corollary~\ref{cor:audit} gives
the entropy precondition under which an extraction-based audit is sound
evidence against a DP claim; Corollary~\ref{cor:blind} shows loss-based
audits have a two-sided blind spot as measurements of memorization,
independent of sample size.

\paragraph{Machine unlearning.}
Exact unlearning by retraining or sharding \cite{cao-yang,bourtoule} and
statistical deletion guarantees \cite{ginart,guo-removal,neel,gupta} led
to the DP-based deletion-capacity theory of \cite{sekhari}, with tight
bounds in \cite{huang-canonne}. For LLMs, approximate unlearning lacks
such guarantees, and benign relearning attacks recover ``forgotten''
verbatim text from loosely related public data \cite{hu-relearn,lucki}.
Corollary~\ref{cor:unlearn} makes ``what was forgotten'' well-posed only
relative to a measure, and our separations exhibit models where
suppressing either standard one is a no-op on the other.

\paragraph{Copyright, axioms, and the long tail.}
Near access-freeness \cite{naf} gives copyright-motivated output bounds
that instantiate the baseline-calibration axiom of
Table~\ref{tab:axioms}, alongside work on whether copyright reduces to
privacy \cite{elkin-koren} and on verbatim memorization as a legal
object \cite{cooper-grimmelmann}; the axiomatic program follows
\cite{kifer-lin}. The long-tail theory of \cite{feldman,brown} explains
\emph{why} memorization occurs, and our plug-in functional damps a
frequent target's counterfactual memorization by exactly $(1-\omega)^n$ in
its corpus mass $\omega$ (Remark~\ref{rem:insupport}).

\begin{table}[t]
\centering\footnotesize
\caption{Theorem~\ref{thm:t2} against the prior $f$-DP bounds on
record-level attack success
\cite{bch22,guo22,hayes-rero,kaissis,kaissis-mech,kulynych}. Row one is
the prior bound; the rows below it are established here.}\label{tab:rero}
\begin{tabular}{@{}lll@{}}
\toprule
 & prior $f$-DP bounds & this paper \\
\midrule
adversary & one reconstruction & adaptive prompting, \\
 & guess & list budget $m$ \\
decoding & not modeled & any decoding rule \\
baseline & unconditional & prior \emph{conditional} on \\
 & & the corpus remainder \\
direction & upper bound & upper bound \emph{and} converse \\
achievability & not characterized & exact: $\kap_0\le\alpha^*(\tau)$ \\
entropy rule & --- & $H_\infty\ge H^*$ \\
\bottomrule
\end{tabular}
\end{table}

\section{Preliminaries and the Game Framework}\label{sec:prelim}

\subsection{Notation and definitions}\label{sec:defs}

\paragraph{Notation.} Vocabulary $V$ ($|V|=v$), documents
$\mathcal{X}=V^{\le L}$, data distribution $\D$, dataset
$S\sim\D^n$. A randomized training algorithm $A:\mathcal{X}^n\to\Theta$
releases a model $\theta$. Trade-off functions and $f$-DP follow
\cite{drs}: $f$ is convex, continuous, nonincreasing, $f(\alpha)\le
1-\alpha$ (hence $f(1)=0$); $A$ is $f$-DP for a neighboring relation if
$T(A(S),A(S'))\ge f$ for all neighbors, both orders. We use replace-one
neighbors for counterfactual statements and add/remove neighbors for
extraction statements, recording conversions where they matter. WLOG $f$
is symmetric ($f=f^{-1}$ in the generalized-inverse sense), as holds for
the three canonical curves: $f_\eps$ (pure $\eps$-DP), $f_{\eps,\delta}$
(approximate DP) and the Gaussian-DP (GDP) curve
$G_\mu(\alpha)=\Phi\big(\Phi^{-1}(1-\alpha)-\mu\big)$, with $\Phi$ the
standard normal CDF. Define the \emph{advantage functional}
\[
\eta(f)\;\eqdef\;\sup_{\alpha\in[0,1]}\big(1-\alpha-f(\alpha)\big),
\]
which equals the total variation $\mathrm{TV}(P,Q)$ when $f=T(P,Q)$ is
attained.

\begin{lemma}[testing form of $f$-DP]\label{lem:c0}
If $A$ is $f$-DP and $P=A(S)$, $Q=A(S')$ for neighbors $S,S'$, then for
every randomized test $\phi:\Theta\to[0,1]$,
$\;\E_P[\phi]\le 1-f(\E_Q[\phi])$.
\end{lemma}

\begin{definition}[counterfactual memorization, plug-in replace-one]
\label{def:mem}
For a bounded score $h:\Theta\times\mathcal{X}\to[0,1]$ and a target $x$
(not necessarily in $\mathrm{supp}(\D)$),
\begin{multline*}
\mem_h(A,\D,n;x)\;\eqdef\;
\E_{S_{-i},\,x'}\Big[\E\,h\big(A(S_{-i}\cup\{x\}),x\big)\\
-\E\,h\big(A(S_{-i}\cup\{x'\}),x\big)\Big],
\end{multline*}
with $S_{-i}\sim\D^{n-1}$, $x'\sim\D$. Appendix~\ref{app:a} relates this
functional to the leave-one-out and population-conditional variants.
\end{definition}

\begin{definition}[extraction game with list budget]\label{def:ext}
Fix $x^*=(c,z)$ with public prefix $c$ and secret suffix
$z\in\mathcal{Z}$. The \emph{conditional prior} $\pi$ is the law of $z$
given the adversary's entire side information, including $c$ and the
realization of the remainder of the training set; paraphrase-correlated
corpus material is thereby charged to $\pi$, not to leakage. An
\emph{extraction protocol} $\mathcal{E}$ with list budget $m$ interacts
with the released model black-box (arbitrarily many adaptive prompts,
each returning the model's next-token conditional distribution at the
queried prefix, with arbitrary decoding) and outputs a list $\Lambda$, $|\Lambda|\le m$; it succeeds iff
$z\in\Lambda$. The protocol is \emph{secret-oblivious}: its prompts,
decoding and list depend on $c$ and the released model but never on $z$,
which is what leaves $\kap_\pi(m)$ as the right null (Appendix~\ref{app:c},
Step 3). The budget is on the list, not on the prompts: a protocol may
decode autoregressively, so one prompt to a served model corresponds to
$|z|$ evaluations of the oracle. Writing $h_z(\theta)\eqdef\Pr[z\in\Lambda\mid\theta]\in[0,1]$,
\[
\Ext(A,\mathcal{E})\;\eqdef\;\E_{z\sim\pi}\,
\E\big[h_z\big(A(S_-\cup\{x^*\})\big)\big],\quad S_-\sim\D^{n-1},
\]
\[
\kap_\pi(m)\;\eqdef\;\sup_{|\Lambda|\le m}\pi(\Lambda)\;\le\;m\cdot 2^{-H_\infty(\pi)}.
\]
$(n,p)$-discoverable extraction by repeated sampling
\cite{carlini-quant,hayes-prob} is the special case of listing $m$
completions. We use the \emph{planted} semantics throughout ($S=S_-\cup\{x^*\}$);
the occupancy-conditioned variant $\E[\,\cdot\mid x^*\in S]$ differs by
corrections of exactly the kind quantified in Appendix~\ref{app:a}.
\end{definition}

\begin{definition}[local score class]\label{def:hloc}
$h\in\Hloc$ iff $h(\theta,x)$ depends on $\theta$ only through the
conditionals at $x$'s own prefixes, $\{p_\theta(\cdot\mid
x_{<t})\}_{t\le|x|}$. Teacher-forcing accuracy, true-token
probability, per-token log-loss, and perplexity all lie in $\Hloc$;
these are the scores used to estimate counterfactual memorization in
practice.
\end{definition}

\subsection{Threat model}\label{sec:threat}

Two positions matter, and one capability is exercised from both of them.
Before the release there is the party that trains and serves; after it,
a party that can only query. The querying party splits not by what it
can do but by what it knows.

\paragraph{Extraction as a capability.}
Definition~\ref{def:ext} fixes it: adaptive black-box prompting under a
list budget, with any decoding rule. The queried party never sees the training pipeline, the optimization
trajectory, or the rest of the corpus, so an observer of training rather
than of the released model is outside the guarantee unless $f$ is the
trajectory-level one. Prior knowledge is charged, not free: everything
known about $z$ beyond the release sits in $\pi$, which is why the
baseline rather than the entropy of $z$ alone appears in the bound.

\paragraph{The adversary, who does not know the secret
(Theorem~\ref{thm:t2}).}
It exercises that capability to \emph{obtain} $z$, and its success is
the harm the bound limits, in expectation over $z\sim\pi$.

\paragraph{The auditor, who does (Corollaries~\ref{cor:audit}
and~\ref{cor:blind}).}
It planted $x^*$, so it knows the target and is measuring rather than
attacking. That extra knowledge buys the loss-based scores: it can
evaluate any bounded $h$ on the released model, and the ones used in
deployment read only the conditionals along $x^*$'s own prefixes, the
class $\Hloc$. It may also run the extraction protocol itself, with knowledge of the
planted secret used only to score the outcome, since $\mathcal{E}$ stays
secret-oblivious. This
is where the two roles meet: an extraction-based audit is the adversary's
capability turned into an instrument, so the same bound that limits the
harm also limits what the audit can conclude
(Corollary~\ref{cor:audit}). Shadow-model attacks such as LiRA need more
than query access, since they train reference models on population data;
we report them separately for that reason.

\paragraph{The trainer or data supplier (Theorem~\ref{thm:t3}).}
Adversarial, and this is the threat model under which memorization
measures are used for compliance: a party who wants a release to pass a
loss-based audit while the content stays recoverable. It may choose the
training data (poisoning), and in the gated release of
Section~\ref{sec:verif-llm} it also controls how the model is served.
Choosing the data and controlling what is fit are not the same power, and
direction~(ii) of Theorem~\ref{thm:t3} needs the second: the separation
collapses once $x^*$ is itself fit (Section~\ref{sec:verif-llm}). The
party it describes controls the run---a malicious trainer, a fine-tuning
service, a compromised supply chain---and one who can only append
documents to someone else's training run is outside it. This is why
worst-case constructions are the right object here, and why we separate
what ordinary fine-tuning produces from what a serving-layer adversary
can build.

Quantitative implication between measures, and separation, are defined in
the natural way (a bound on one transfers to the other with an explicit
loss; or an explicit construction defeats every such transfer).
Table~\ref{tab:map} summarizes the games and the bridges proved below.

\begin{table*}[t]
\small
\centering
\caption{The measures as adversary games and their quantitative bridges
from $f$-DP. The red edge of Fig.~\ref{fig:map}
(Theorem~\ref{thm:t3})
separates rows three and five in both directions on $\Hloc$.}\label{tab:map}
\begin{tabular}{@{}llll@{}}
\toprule
Measure & Adversary game & Bridge from $f$-DP & Tightness \\
\midrule
MI advantage & membership test & exact region \cite{yeom,drs} & known \\
Reconstruction & ReRo game & prior-baseline bound \cite{bch22} & known \\
Counterfactual mem. & score gap on $\Hloc$ & $\eta(f)$
  (Thm.~\ref{thm:t1}) & tight, within $\Hloc$ \\
$k$-copy mem. & $k$-hop advantage & $\eta^*_k$
  (Lem.~\ref{lem:staircase}) & exact; $<$ group bound \\
Extraction & adaptive, list budget $m$ & $1-f(\kap_\pi(m))$
  (Thm.~\ref{thm:t2}) & tight on dense $\kap$ \\
Distributional mem. & corpus correlation & --- & open
  (Conj.~\ref{conj:axioms}) \\
\bottomrule
\end{tabular}
\end{table*}

\section{Calibration: $f$-DP Bounds Counterfactual Memorization}
\label{sec:t1}

The first bridge is a cap. Whatever bounded score an auditor computes,
$f$-DP limits how much of it the presence of one record can buy, and the
limit is a single functional of $f$ with closed forms for the guarantees
in use.

\begin{theorem}\label{thm:t1}
If $A$ is $f$-DP under replace-one, then for all $\D,n,x$ and all bounded
scores $h$: $\;|\mem_h|\le\eta(f)$, with
$\eta(f_\eps)=\tanh(\eps/2)$,
$\eta(f_{\eps,\delta})=\delta+(1-\delta)\tanh(\eps/2)$,
$\eta(G_\mu)=2\Phi(\mu/2)-1$. The bound is tight, attainable by a score in
$\Hloc$, and fails for unbounded scores.
\end{theorem}

The proof (Appendix~\ref{app:b}) is a pointwise application of
Lemma~\ref{lem:c0} in both orders; Fig.~\ref{fig:eta} plots the three
advantage curves $\alpha\mapsto1-\alpha-f(\alpha)$ with their maximizers.
Tightness uses a membership-binary canary released as the
pushforward of a pair attaining $f$, realized as a language model whose
probability of emitting \texttt{1} at $c$ is the optimal test evaluated
on the released draw, so the attaining score is the
teacher-forcing probability, an element of $\Hloc$
(Remark~\ref{rem:hloc-tight}). Boundedness is necessary: with log-loss the
memorization of an $\eps$-DP randomized-response canary is
$\tanh(\eps/2)\log(t/s)\to\infty$; truncation at level $M$ prices in as
$M\tanh(\eps/2)$, exactly.

\begin{figure}[t]
\centering
\includegraphics[width=0.7\columnwidth]{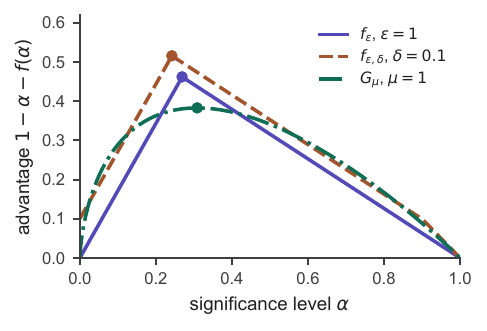}
\caption{Advantage curves $1-\alpha-f(\alpha)$ for the three canonical
guarantees. The marked maximizers carry the closed forms of
Theorem~\ref{thm:t1}: $\tanh(\eps/2)$ at $1/(e^{\eps}{+}1)$,
$\delta+(1-\delta)\tanh(\eps/2)$ at $(1-\delta)/(e^{\eps}{+}1)$, and
$2\Phi(\mu/2)-1$ at $\Phi(-\mu/2)$.}
\label{fig:eta}
\end{figure}

\begin{remark}[tightness within $\Hloc$]\label{rem:hloc-tight}
Since the attaining score lies in $\Hloc$, Theorem~\ref{thm:t1} is tight
even when scores are restricted to the local class, so the constant is
attained on the same class the separations of
Section~\ref{sec:t3} use.
\end{remark}

\begin{corollary}[duplication upper bound]\label{cor:dup}
If $x$ occurs as $k$ copies, group privacy gives
$|\mem^{(k)}_h|\le\eta(f^{(k)})$ for the exact $k$-fold trade-off
$f^{(k)}=1-(1-f)^{\circ k}$ of \cite{drs}; the naive $\eps\mapsto k\eps$
conversion gives the weaker $\tanh(k\eps/2)$.
\end{corollary}

Writing $\bar f\eqdef1-f_\eps$, the advantage functional of the exact
trade-off is $\eta(f^{(k)})=\sup_x\big(\bar f^{\circ k}(x)-x\big)$: a
one-dimensional staircase optimum. Lemma~\ref{lem:staircase} solves it in
closed form for every $k$, exhibits a mechanism attaining it, and places
it strictly below $\tanh(k\eps/2)$ for $k\ge2$.

\begin{lemma}[exact $k$-copy constant, pure DP]\label{lem:staircase}
The maximum bounded-score advantage of an $\eps$-DP mechanism across
$k$-hop neighboring datasets equals
$\eta^*_k(\eps)=\max_{x\in[0,1]}\big(\bar f^{\circ k}(x)-x\big)$ with
$\bar f(b)=\min(e^{\eps}b,\,1-e^{-\eps}(1-b))$, attained by a
Bernoulli-staircase mechanism on the copy count that is $\eps$-DP on all
datasets, with the attaining score in $\Hloc$. In closed form,
\[
\eta^*_k(\eps)=
\begin{cases}
1-e^{-k\eps/2}, & k \text{ even},\\[3pt]
1-\dfrac{2\,e^{-(k-1)\eps/2}}{e^{\eps}+1}, & k \text{ odd},
\end{cases}
\]
recovering $\tanh(\eps/2)$ at $k=1$. \emph{Every} case is attained
exactly by the geometric mechanism on the count: its shift-$k$ total
variation equals the closed form for all $k$ (the continuous Laplace
mechanism, with $\mathrm{TV}=1-e^{-k\eps/2}$, matches only even $k$).
Moreover $\eta^*_k<\tanh(k\eps/2)$ strictly for all $k\ge2$ (e.g.\
$0.632<0.762$ at $k=2$, $\eps=1$). Proof in
Appendix~\ref{app:b-stair}.
\end{lemma}

Deduplication guidance thus reads: the admissible memorization advantage
of a $k$-duplicated canary grows as $\eta^*_k$, not as the naive
$\eps\mapsto k\eps$ bound, and noisy counting is the extremal mechanism;
Fig.~\ref{fig:stair} shows the optimal orbit and the gap to the group
bound.

\begin{figure}[t]
\centering
\includegraphics[width=\columnwidth]{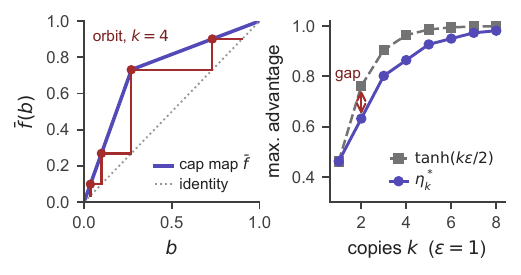}
\caption{Lemma~\ref{lem:staircase}. \emph{Left:} the cap map
$\bar f(b)=\min(e^{\eps}b,\,1-e^{-\eps}(1-b))$ and the
optimal $k=4$ orbit (cobweb): exponential growth up to the breakpoint
$1/(e^{\eps}{+}1)$, then geometric approach to $1$. \emph{Right:} the
exact constant $\eta^*_k$ versus the unattainable naive bound
$\tanh(k\eps/2)$; the $k=2$ gap is $0.762-0.632$ at $\eps=1$.}
\label{fig:stair}
\end{figure}

\begin{remark}[per-instance refinement]\label{rem:per-instance}
Replacing $f$ by a per-instance trade-off $f_x$ carries the proof through
verbatim: $|\mem_h(x)|\le\eta(f_x)$, which predicts which examples
will be memorized from their per-instance privacy parameters, a
prediction our experiments confirm directly.
\end{remark}

The functional $\eta(f)$ is the total-variation parameter implied by
$f$-DP \cite{kaissis-mech}, its counterfactual-memorization instance is
\cite[Prop.~F.1]{kulynych}, and the $\eps$-DP label case is in
\cite{feldman}. Theorem~\ref{thm:t1} states it in the plug-in
replace-one semantics used throughout. Attainment \emph{within} $\Hloc$,
the truncation price for unbounded scores, and the exact $k$-copy
constant of Lemma~\ref{lem:staircase} are established here.

\section{The Extraction Boundary}\label{sec:t2}

\subsection{The bound and its converse}\label{sec:t2-main}

The second bridge is a boundary rather than a cap. Part (i) bounds
extraction the way Theorem~\ref{thm:t1} bounds memorization, but the
bound alone would not say when DP is enough; parts (ii) and (iii) close
it from the other side, and the two together turn a guarantee
into a condition an auditor can check.

\begin{theorem}\label{thm:t2}
Let $A$ be $f$-DP under add/remove neighbors, in the game of
Definition~\ref{def:ext}.
\textbf{(i)} For every protocol $\mathcal{E}$,
$\;\Ext(A,\mathcal{E})\le 1-f\big(\kap_\pi(m)\big)$. In particular, with
$\kap=\kap_\pi(m)$: pure $\eps$-DP gives
$\Ext\le\min\!\big(1,\,e^{\eps}\kap,\,1-e^{-\eps}(1-\kap)\big)\le
e^{\eps}\kap$; $(\eps,\delta)$-DP gives $\Ext\le\delta+e^{\eps}\kap$;
$\mu$-GDP gives $\Ext\le\Phi\big(\Phi^{-1}(\kap)+\mu\big)$.
\textbf{(ii)} For every baseline of the form $\kap=m/N$ (dense in
$(0,1)$), the subset-release mechanism of Lemma~\ref{lem:gadget} is $f$-DP
on all datasets and admits a protocol with $\Ext=1-f(\kap)$.
\textbf{(iii)} Fix $\tau\in(0,1)$ with $\tau\ge 1-f(0)$ (for
$(\eps,\delta)$-DP: $\tau\ge\delta$, necessary), and let
$\alpha^*(\tau)=\sup\{\alpha:1-f(\alpha)\le\tau\}$. Every $f$-DP algorithm
satisfies $\Ext\le\tau$ for every prior--protocol pair with
$\kap_\pi(m)\le\kap_0$ \emph{if and only if} $\kap_0\le\alpha^*(\tau)$; the pair
carries its own list budget (Definition~\ref{def:ext}), so the
quantifier ranges over budgets as well as priors. For pure $\eps$-DP,
$\alpha^*(\tau)=e^{-\eps}\tau$ if $\tau\le e^{\eps}/(e^{\eps}+1)$ and
$\alpha^*(\tau)=1-e^{\eps}(1-\tau)$ otherwise; in the first regime
(which contains every $\tau\le1/2$, since $e^{\eps}/(e^{\eps}+1)>1/2$)
the condition takes an entropy form. Write
\[
H^*\eqdef\eps\log_2 e+\log_2(m/\tau).
\]
\emph{Sufficiency is universal:} every prior with $H_\infty(\pi)\ge H^*$
satisfies $\Ext\le\tau$ against every $f$-DP algorithm, via the
dictionary $\kap_\pi(m)\le m2^{-H_\infty}$, which is valid for all
$\tau$. \emph{The threshold is exact on the uniform family:} a uniform
prior meets the dictionary with equality, so it is safe against every
$\eps$-DP algorithm \emph{if and only if} $H_\infty\ge H^*$, the
necessity being witnessed by the gadget of Lemma~\ref{lem:gadget} at
$\kap=m2^{-H_\infty}>\alpha^*$. For
$\tau>e^{\eps}/(e^{\eps}+1)$ the entropy form must be stated with
$\alpha^*(\tau)=1-e^{\eps}(1-\tau)$, which exceeds $e^{-\eps}\tau$: the
naive formula's necessity direction fails there
(Remark~\ref{rem:naive-fails} gives a counterexample). For an individual
non-uniform prior $\kap_\pi(m)\le\alpha^*(\tau)$ remains sufficient and
can hold at strictly smaller min-entropy, the dictionary being loose by
a factor up to $m$; whether it is also necessary at a fixed non-uniform
profile is open (Remark~\ref{rem:general-kappa}).
\end{theorem}

\begin{lemma}[subset-release gadget]\label{lem:gadget}
Let $\pi$ be uniform on $\{1,\dots,N\}$, $1\le m<N$, $f$ symmetric, and
$a=\frac{1-f(m/N)-m/N}{1-m/N}\in[0,1]$. Release, when the dataset contains
\emph{exactly one} prefix-$c$ document $(c,z)$, a size-$m$ candidate set
$C\sim a\,U_{\ni z}+(1-a)\,U$; otherwise release $C\sim U$. Here $U$ is
uniform on the size-$m$ subsets of $\{1,\dots,N\}$ and $U_{\ni z}$
uniform on those containing $z$. Then every
add/remove neighbor pair of output laws is of the form $(P_z,Q)$ or
identical (pairs $(P_z,P_{\tilde z})$ never arise), and
$T(P_z,Q)\ge f$, $T(Q,P_z)\ge f$; on the planted instance,
$\Ext=1-f(m/N)$ with $\kap_\pi(m)=m/N$.
\end{lemma}

Proofs are in Appendix~\ref{app:c}. Four remarks: part (i) consumes only the \emph{ordered} guarantee
$T\big(A(S_-),A(S_-\cup\{x^*\})\big)\ge f$, with the baseline as the
null, so it holds under one-sided (add-only) relaxations of $f$-DP;
the bound holds verbatim under replace neighbors and, more generally, with
corpus-dependent priors upon replacing $\kap$ by
$\bar\kap=\E_{S_-}[\sup_{|\Lambda|\le m}\pi(\Lambda\mid S_-)]$; canary-based auditing
is sound only when the canary's conditional entropy clears the threshold
of (iii) (Corollary~\ref{cor:audit}); and duplication combines with entropy as
$H_\infty\gtrsim k\eps\log_2 e+\log_2(m/\tau)$ via group privacy.

\begin{corollary}[audit interface]\label{cor:audit}
An extraction-based audit refutes an $f$-DP claim only if the planted
canary satisfies $\kap_\pi(m)\le\alpha^*(s)$, where $s$ is the
population extraction success or, in a finite-sample audit, a valid
lower confidence bound on it; low-entropy or corpus-correlated canaries, whose baseline exceeds
$\alpha^*(s)$, yield successes that are consistent with the claimed
privacy. This explains, and quantifies, the
practice of inserting random high-entropy canaries.
\end{corollary}

\begin{remark}[replace neighbors: a bracket]\label{rem:bracket}
Under replace, same-prefix pairs $(P_z,P_{\tilde z})$ are binding for the
converse; instantiating the gadget with the group square root $f^{1/2}$
(pure DP: $f_{\eps/2}$) makes all replace pairs satisfy $f$ and yields
violations for $\kap>\alpha^*_{f^{1/2}}(\tau)$ (pure DP:
$e^{-\eps/2}\tau$). The exact replace threshold inside
$(\alpha^*_f(\tau),\,\alpha^*_{f^{1/2}}(\tau)]$ is open;
Remark~\ref{rem:bracket-app} gives the argument.
\end{remark}

\subsection{Reconciliation with the extraction--indistinguishability
separation}\label{sec:reconcile}

\cite{lex26} proves that indistinguishability is neither sufficient nor
necessary for inextractability. Our boundary speaks to the
\emph{sufficiency} direction, and it does so deductively rather than by
inspection of their examples:

\begin{corollary}\label{cor:forced}
Any mechanism that is $f$-DP in the game of Definition~\ref{def:ext} and
achieves $\Ext>\tau$ necessarily has
$\kap_\pi(m)>\alpha^*(\tau)$.
\end{corollary}

\begin{proof}
Contrapositive of Theorem~\ref{thm:t2}(i).
\end{proof}

Thus \emph{every} witness of ``DP yet extractable'', theirs or any
other, is forced into the low-entropy / high-baseline region
$\kap>\alpha^*(\tau)$, once its game embeds into
Definition~\ref{def:ext}. For the $(l,b)$-inextractability game of
\cite{lex26} (at least $2^{b}$ expected black-box queries to induce
emission of a protected $l$-gram), the two notions budget different
resources. Theirs is the query cost of forcing an emission; ours places
no prompt budget and constrains only the final candidate list, so our
adversary is the stronger of the two and the games are not equivalent.
The comparison locates a common targeted-extraction regime. What
transfers is the
direction of the baseline: their target is an $l$-gram \emph{substring}
emittable in any context, whereas our game fixes the prefix, and the
fixed-prefix event is contained in the any-context one, so their
baseline dominates $\kap$ and any witness that does embed lands only
\emph{deeper} inside the forced region. The two results are
therefore not in tension: Corollary~\ref{cor:forced} says \emph{where}
indistinguishability stops controlling extraction, and the separation of
\cite{lex26} exhibits inhabitants of that region; their
\emph{necessity}-direction counterexamples (inextractable yet
distinguishable) concern the converse implication, which neither our
framework nor theirs asserts. Fig.~\ref{fig:boundary} plots the
boundary $1-f(\kap)$, the target risk $\tau$, and the forced region.

\begin{figure}[t]
\centering
\includegraphics[width=0.8\columnwidth]{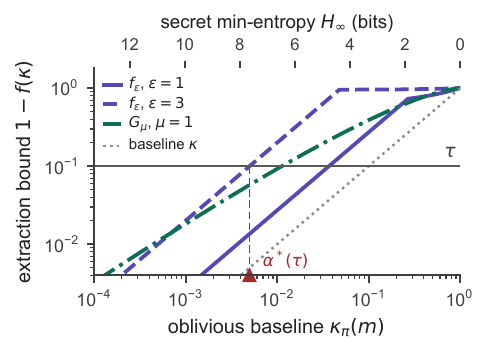}
\caption{The extraction boundary $1-f(\kap)$ against the oblivious
baseline (log--log), with the top axis reading the baseline as secret
min-entropy. The horizontal line is a target risk $\tau=0.1$, and the
marker on the baseline axis is $\alpha^*(\tau)$, carried by the dashed
rule up to the $\eps=3$ crossing. Everything to its right has
$1-f(\kap)>\tau$, so by Corollary~\ref{cor:forced} every ``DP yet
extractable'' witness has a baseline there; to its left
$H_\infty\ge H^*$ and DP holds extraction below $\tau$.}
\label{fig:boundary}
\end{figure}

\section{Separations on the Local Score Class}\label{sec:t3}

Fixing $\Hloc$ is necessary for the question to be nontrivial: a score
may not itself run an extraction protocol.

\begin{theorem}\label{thm:t3}
\textbf{(i) Memorized but unextractable.} For every $2\le k<v$ and $L$
there are model laws $(P_z,Q)$ (trained with/without $x^*$) with
teacher-forcing $\mem=\tfrac1k-\tfrac1v$ for the unsmoothed variant
(and, for the $\gamma=1/v$-smoothed variant of the same construction,
normalized log-score $\mem\to1$ as $v\to\infty$), yet
$\Ext\le m\,k^{-L}$ against \emph{every} adaptive white-box protocol.
\textbf{(ii) Extractable but locally invisible.} There is a construction
with $\Ext=1$ at list budget $m=1$ and $\kap_\pi(1)=v^{-L}\le 2^{-L}$, while
$\mem_h=0$ \emph{exactly} for every $h\in\Hloc$.
\textbf{(iii)} Hence on $\Hloc$ neither measure controls the other; the
direction-(i) gap grows exponentially in $L$, and the direction-(ii) gap
is maximal outright.
\end{theorem}

Direction (i) is the \emph{uniform camouflage tree}
(Fig.~\ref{fig:constructions}, left): on-path conditionals
are uniform over a $k$-set containing the true token among random decoys;
off-path over random $k$-sets. The key facts (Appendix~\ref{app:d}) are
posterior uniformity (given the entire model, the secret is uniform
over exactly $k^L$ surviving paths) and the \emph{perfect camouflage
identity} $\E_z[P_z]=Q$: the marginal model law is independent of the
secret, so distribution-level audits are blind to direction-(i)
memorization in principle, not merely in samples. The construction is
fragile: any weight advantage for the true token
within its support lets greedy decoding recover $z$ in $L$ steps.
Membership in the support must be the only signal, drowned among decoys.
Direction (ii) is a reserved-trigger side channel
(Fig.~\ref{fig:constructions}, right): the model equals the
base model on every prefix of $x^*$ (hence on everything any local score
reads) and deterministically emits $z$ after a reserved trigger $\rho$
that appears on no prefix of $x^*$; divergence-style attacks \cite{nasr23} and poisoning
backdoors \cite{badnets,wan} are its instances in the wild.

\begin{figure}[t]
\centering
\begin{tikzpicture}[font=\small, level distance=7mm,
  level 1/.style={sibling distance=10mm},
  level 2/.style={sibling distance=5mm},
  surv/.style={circle, draw, fill=hsTeal!12, inner sep=1.6pt},
  dead/.style={circle, draw=hsGray!60, inner sep=1.6pt},
  tok/.style={circle, draw, font=\scriptsize, inner sep=0pt,
              minimum size=4.6mm},
  lab/.style={font=\footnotesize}]
\begin{scope}
\node[surv] (r) {$c$}
  child {node[surv] {} child {node[surv] (lf1) {}} child {node[surv] (lf2) {}}
         child[missing] {} }
  child {node[surv] {} child {node[surv] (lf3) {}} child {node[surv] (lf4) {}} }
  child {node[dead] {} edge from parent[hsGray!60, dashed]}
  child {node[dead] {} edge from parent[hsGray!60, dashed]};
\draw[hsRust, line width=0.7pt] (lf3) circle (0.19);
\node[font=\scriptsize, text=hsRust, anchor=east] at ($(lf3)+(-0.26,0.02)$) {$z$};
\draw[decorate, decoration={brace, amplitude=3.5pt, mirror}, hsGray!70]
  ($(lf1)+(-0.19,-0.24)$) -- ($(lf4)+(0.19,-0.24)$);
\node[font=\footnotesize, text=black!60, anchor=north]
  at ($(lf1)!0.5!(lf4)+(0,-0.36)$) {$k^L$ equiprobable};
\node[lab, anchor=north, align=center] at (0,-2.15)
  {I: $\mem=\tfrac1k-\tfrac1v$,\\$\Ext\le mk^{-L}$};
\end{scope}
\begin{scope}[xshift=3.86cm]
\node[draw, rounded corners=2pt, fill=hsGray!10, minimum width=2.2cm,
      minimum height=0.6cm] (B) at (0,0) {base model $B$};
\node[tok, fill=hsRust!12] (t0) at (-1.82,-1.30) {$\rho$};
\node[tok, fill=hsRust!12] (t1) at (-1.16,-1.30) {$z_1$};
\node[tok, fill=hsRust!12] (t2) at (-0.50,-1.30) {$z_2$};
\draw[-{Stealth}, hsRust] (t0) -- (t1);
\draw[-{Stealth}, hsRust] (t1) -- (t2);
\draw[hsRust!75] (B.south) to[out=250,in=90] (t0.north);
\node[font=\scriptsize, text=hsRust, anchor=north] at (-1.16,-1.58) {rewritten};
\node[tok, fill=hsGray!8, draw=hsGray!55] (c0) at (0.50,-1.30) {$c$};
\node[tok, fill=hsGray!8, draw=hsGray!55] (c1) at (1.16,-1.30) {$z_1$};
\node[tok, fill=hsGray!8, draw=hsGray!55] (c2) at (1.82,-1.30) {$z_2$};
\draw[-{Stealth}, hsGray!65] (c0) -- (c1);
\draw[-{Stealth}, hsGray!65] (c1) -- (c2);
\draw[hsGray!65] (B.south) to[out=290,in=90] (c0.north);
\node[font=\scriptsize, text=black!58, anchor=north] at (1.16,-1.58) {unchanged};
\node[lab, anchor=north, align=center] at (0,-2.12)
  {II: $\mem_h=0$ on $\Hloc$,\\$\Ext=1$};
\end{scope}
\end{tikzpicture}
\caption{The two constructions of Theorem~\ref{thm:t3} ($v=4$, $k=2$,
$L=2$ shown). \emph{Left:} the true path $z$ is circled, pruned branches
are dashed, and the brace marks the surviving paths, over which the
posterior is uniform (Lemma~\ref{lem:tree-post}). \emph{Right:} the
rewritten trigger subtree, and beside it the document's own prefixes,
which the construction leaves identical to the base model.}
\label{fig:constructions}
\end{figure}
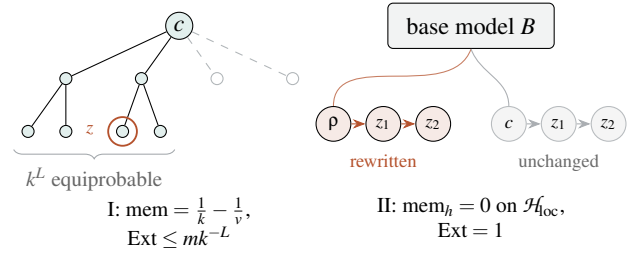

\begin{corollary}[two-sided audit blind spot]\label{cor:blind}
Loss-based membership or unlearning verification systematically
\emph{misses} direction-(ii) leakage and systematically \emph{flags}
direction-(i) memorization as if it were extraction risk; the two halves
are Theorem~\ref{thm:t3}(ii) and (i) read as audit statements.
\end{corollary}

\begin{corollary}[unlearning semantics]\label{cor:unlearn}
``What was forgotten'' must be asserted relative to a measure: suppressing
extraction is a no-op on direction-(i) models, and suppressing document
loss is a no-op on direction-(ii) models.
\end{corollary}

\begin{corollary}[post-processing axis]\label{cor:pp}
Post-processing separates the two extraction semantics. The
\emph{sup-protocol} extractability of Definition~\ref{def:ext} is
post-processing monotone (Proposition~\ref{prop:axioms}f). The
\emph{fixed-protocol verbatim-emission} measures used in practice
(greedy prefix completion, $(n,p)$-discoverable extraction with the
document's own prefix \cite{carlini-quant,hayes-prob}) are not: for
Construction~II there are data-independent kernels moving the
fixed-prefix greedy emission rate from $v^{-L}$ to $1$ (the copy kernel
of Proposition~\ref{prop:axioms}b) and the trigger emission rate from
$1$ to $v^{-L}$ (randomizing the $\rho$-subtree). In each case a counterfactual
measure certifies that the underlying information is untouched, but it is
a \emph{different} one each time: across the copy kernel the sup-score
functional is unchanged at exactly $1$ (the indicator of the intact
trigger subtree separates the branches perfectly, before and after),
while the $\Hloc$-restricted functional jumps per
Proposition~\ref{prop:axioms}(b); across the trigger-randomizing kernel
the $\Hloc$-restricted functional is unchanged (exactly $0$), while the
sup-score functional drops from exactly $1$ to $0$, a decrease
consistent with Proposition~\ref{prop:axioms}(a). Proof in
Appendix~\ref{app:e}.
\end{corollary}

\begin{remark}[status of the two constructions]\label{rem:artificial}
Both constructions are explicit mechanisms rather than SGD products, but
their status differs. Construction~II is
\emph{realizable against real SGD pipelines}: a data supplier who plants
trigger-completion pairs implements exactly the reserved-subtree channel,
and instruction-tuning poisoning succeeds with a handful of examples
\cite{wan,badnets}. Our experiments confirm the record-level
counterfactual under ordinary fine-tuning, though only its extraction
side exactly; the separation additionally needs the released model not to
be fit on $x^*$, so the party controls the run rather than merely
appending to it (Section~\ref{sec:threat}). Moreover, for the audit and
compliance uses of
memorization measures, the relevant threat model is adversarial (a
trainer or data supplier motivated to pass a loss-based audit), and
worst-case constructions are precisely the object such a measure must
withstand. Construction~I raises the genuinely open question of whether
\emph{natural} training produces camouflage-type geometry; we sharpen it
into a testable hypothesis: train a small language model on data drawn
from the camouflage distribution and measure the posterior flatness of
the learned conditionals over surviving candidates.
A flat learned posterior would show SGD can
\emph{represent} the geometry when the data carries it, a weaker but
verifiable claim, distinct from spontaneous emergence, which we leave
open. Finally, direction (ii)'s invisibility is relative to $\Hloc$: the
correct reading is that local scores are insufficient, not that the
counterfactual notion is unsalvageable.
\end{remark}

\section{Axioms for Memorization Measures}\label{sec:axioms}

Following the axiomatic precedent of \cite{kifer-lin} for statistical
privacy, we propose four axioms for a memorization measure $\mathcal{M}$ mapping a
(pipeline, target) pair to $[0,1]$:
\begin{itemize}[leftmargin=*, itemsep=2pt]
\item[\textbf{P}] (post-processing) $\mathcal{M}(\Psi\circ A;x)\le\mathcal{M}(A;x)$ for
every data-independent Markov kernel $\Psi$ on models;
\item[\textbf{C}] (composition) joint release of $A_1,\dots,A_r$ degrades
$\mathcal{M}$ in a controlled way, e.g.\ via the composed guarantee
$f_1\otimes\cdots\otimes f_r$;
\item[\textbf{B}] (baseline calibration) content reproducible by a
data-independent mechanism scores $0$ (the ReRo prior term
\cite{bch22} and near access-freeness \cite{naf} instantiate this);
\item[\textbf{E}] (estimability) there is an estimator that, given
black-box query access to the \emph{single released model}
(polynomially many queries, no access to the training pipeline, to
retraining, or to auxiliary population data) is consistent for $\mathcal{M}$ \emph{uniformly over
pipelines}: for every $\delta>0$ it outputs $\mathcal{M}(A;x)\pm\delta$ with
probability $1-\delta$ for every pipeline $A$ in its scope.
\end{itemize}
The access model in \textbf{E} is deliberately strict; estimability under
enlarged access (shadow-model training on population data, or pipeline
access with retraining) is recorded separately in
Table~\ref{tab:axioms}.

The separation constructions settle several cells outright.

\begin{proposition}\label{prop:axioms}
\emph{(a)} The sup-score counterfactual functional
$\overline{\mem}(A;x)\eqdef\sup_{h\in[0,1]}|\mem_h(A;x)|$ satisfies
\textbf{P}. \emph{(b)} Its $\Hloc$-restricted version violates
\textbf{P}: there is a fixed kernel $\Psi$ that raises the
$\Hloc$-restricted measure of Construction~II from exactly $0$ to at
least $1-1/v$.
\emph{(c)} Extraction and both counterfactual functionals satisfy
\textbf{C}: jointly releasing $r$ models with guarantees
$f_1,\dots,f_r$ gives
$\Ext\le1-(f_1\otimes\cdots\otimes f_r)(\kap_\pi(m))$ and
$|\mem_h|\le\eta(f_1\otimes\cdots\otimes f_r)$ for every bounded $h$.
\emph{(d)} Raw extraction violates \textbf{B}: a data-independent
mechanism already emits any $\Theta(1)$-prior string, so
verbatim-emission measures charge prior-guessable content. \emph{(e)}
Neither counterfactual functional, nor the membership-inference
advantage, satisfies \textbf{E}: there are pipelines releasing identical
models on the planted instance whose counterfactual functionals differ by
$1-o(1)$ and whose MI advantages differ by $1$. \emph{(f)} The
sup-protocol extractability
$\Ext^{\sup}(A)\eqdef\sup_{\mathcal{E}}\Ext(A,\mathcal{E})$ satisfies
\textbf{P}: $\Ext^{\sup}(\Psi\circ A)\le\Ext^{\sup}(A)$ for every
data-independent kernel $\Psi$.
\end{proposition}

Proofs are in Appendix~\ref{app:e}; (e) rests on the observation that
these are \emph{two-branch} quantities: they compare the branch of the
pipeline that saw $x^*$ with the branch that did not, while black-box
access touches only the realized branch, and the argument applies to the
MI advantage verbatim. Part (b) is the axiomatic face of
Theorem~\ref{thm:t3}(ii): restricting to $\Hloc$ is what buys black-box
estimability of the defining \emph{proxy score} (the functional itself
remains inestimable by (e)), and the same restriction is what breaks
\textbf{P}. This trade surface is the evidence for the following.

\begin{conjecture}\label{conj:axioms}
No memorization measure that is not identically zero simultaneously
satisfies \textbf{P}, \textbf{B}, and \textbf{E}, with \textbf{E} in the
strict single-release access model defined above.
\end{conjecture}

\begin{table*}[t]
\centering\footnotesize
\caption{Axiom satisfaction. The access model for \textbf{E} is black-box
query access to the single released model; ``proxy'' marks estimability
of the defining single-model \emph{score} (what practice measures
\cite{zhang-cf}), and ``w/ shadows'' marks estimability under the
enlarged access of shadow-model training on population data
\cite{shokri,lira}. Cells carrying Prop./Thm./Cor.\ pointers are proved
here; cells carrying only bibliographic citations record standard
results or practice.}\label{tab:axioms}
\begin{tabular}{@{}lcccc@{}}
\toprule
Measure & P & C & B & E \\
\midrule
Extraction (sup-protocol) & $\checkmark$ (Prop.~\ref{prop:axioms}f) &
$\checkmark$ (Prop.~\ref{prop:axioms}c) & $\times$
(Prop.~\ref{prop:axioms}d) & open \\
Verbatim emission (fixed protocol) & $\times$ (Cor.~\ref{cor:pp}) &
$\checkmark$ (Prop.~\ref{prop:axioms}c) & $\times$
(Prop.~\ref{prop:axioms}d) & $\checkmark$ \cite{hayes-prob} \\
Counterfactual, sup-score & $\checkmark$ (Prop.~\ref{prop:axioms}a) &
$\checkmark$ (Prop.~\ref{prop:axioms}c) & open & $\times$
(Prop.~\ref{prop:axioms}e) \\
Counterfactual on $\Hloc$ & $\times$ (Prop.~\ref{prop:axioms}b) &
$\checkmark$ (Prop.~\ref{prop:axioms}c) & open & $\times$
(Prop.~\ref{prop:axioms}e); proxy $\checkmark$ \cite{zhang-cf} \\
MI advantage & $\checkmark$ (data processing; \cite{drs}) & $\checkmark$
\cite{kov,drs} & $\checkmark$ \cite{yeom} & $\times$
(Prop.~\ref{prop:axioms}e); w/ shadows $\checkmark$ \cite{shokri,lira} \\
Distributional mem. & open & open & open & open \\
\bottomrule
\end{tabular}
\end{table*}

No row of Table~\ref{tab:axioms} satisfies \textbf{P}, \textbf{B} and
\textbf{E} at once under the strict access model. The MI row carries
four ticks and might look like a counterexample; Appendix~\ref{app:e-mi} explains why it is not,
and sketches the route by which \textbf{B} and single-release
consistency appear to collide.
Table~\ref{tab:axioms} records the proved cells; the remainder is open.

\section{Experiments}\label{sec:verif}

The experiments run at two scales. Four mechanism-level experiments
exercise sampled corpora and small transformers trained by ordinary SGD,
where every quantity in the theorems can be measured directly against its
closed form. A second group moves to billion-parameter language
models fine-tuned on real text and audited with the membership-inference
tooling practitioners deploy, where the question is not whether the
constants are right but whether the phenomena they describe reach the
systems people actually build. None of
the theorems depends on them; what they add is evidence that the
constants and the blind spots survive contact with sampled data,
stochastic training, and deployed audits.

\subsection{Mechanism-level experiments}\label{sec:verif-mech}

Figs.~\ref{fig:exp1}--\ref{fig:exp4} report one experiment each.

\begin{figure}[t]
\centering
\includegraphics[width=0.9\columnwidth]{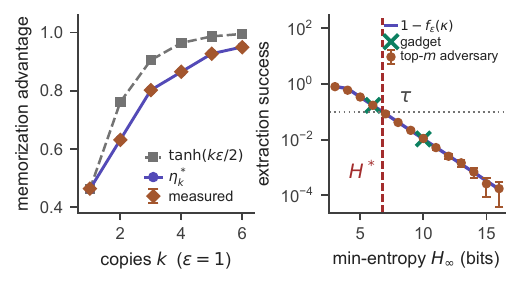}
\caption{A canary planted $k$ times in a Zipf corpus,
released as geometric noisy counts ($\eps=1$). \emph{Left:} the measured
best-threshold advantage (Monte-Carlo error bars) attains the exact
constant $\eta^*_k$ of Lemma~\ref{lem:staircase} and stays strictly below
the group bound $\tanh(k\eps/2)$. \emph{Right:} an entropy sweep over
\emph{uniform} priors, where $\kap_\pi(m)=m2^{-H_\infty}$ exactly, so
the plotted curve is the frontier $1-f_\eps(\kap)$ of
Theorem~\ref{thm:t2}. Dots are the top-$m$ list adversary ($m=4$) and
crosses the subset-release gadget of Lemma~\ref{lem:gadget}; the dotted
rule is the target risk $\tau$ and the dashed one the threshold $H^*$.}
\label{fig:exp1}
\end{figure}

\begin{figure}[t]
\centering
\includegraphics[width=0.8\columnwidth]{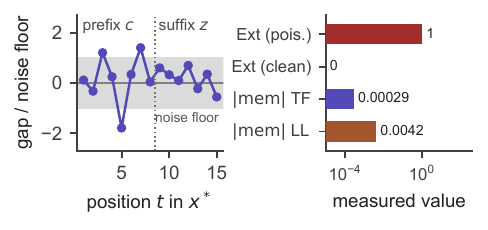}
\caption{Construction~II realized by SGD. A two-layer
transformer is poisoned with $64$ trigger pairs ($0.6\%$ of a Markov--Zipf
corpus), three seeds per condition. \emph{Left:} on $x^*$'s own prefixes, the
poisoned$-$clean teacher-forcing gap, in units of the seed-to-seed noise
floor, stays inside the $\pm1$ band: the poisoning gap is
indistinguishable from the disagreement between two clean models.
\emph{Right:} extraction rates on the poisoned and clean runs, and
$|\mem_h|$ for the two local scores, teacher-forcing (TF) and log-loss
(LL), on a log axis.}
\label{fig:exp2}
\end{figure}

\paragraph{The staircase and the boundary
(Lemma~\ref{lem:staircase}, Theorem~\ref{thm:t2}).}
Releasing geometric noisy counts of a canary planted $k$ times in a Zipf
corpus, the measured best-threshold advantage matches $\eta^*_k$ to
Monte-Carlo precision for every $k\le6$ and sits strictly below the
group bound for $k\ge2$ ($0.631$ vs $0.762$ at $k=2$;
Fig.~\ref{fig:exp1}, left). The same pipeline's top-$m$ list adversary
tracks the boundary across thirteen octaves of prior entropy, crossing
the risk level $\tau=0.1$ exactly at the predicted threshold
$H^*=6.77$ bits. That sweep runs over uniform priors, so
$\kap_\pi(m)=m2^{-H_\infty}$ and the curve is the exact frontier rather
than the entropy envelope, and the gadget of
Lemma~\ref{lem:gadget} sits on it at dense baselines
(Fig.~\ref{fig:exp1}, right). Under the same mechanism and the same
$\eps$, a low-entropy Zipf-prior canary ($H_\infty=2.5$ bits) is
extracted $41\%$ of the time \emph{without violating} the $\eps$-DP
claim; the same release holds a $10$-bit uniform canary to $1\%$. This
is the audit scenario that Corollary~\ref{cor:audit} formalizes: a
canary success that refutes nothing. It is also where the baseline and
its entropy certificate come apart: the Zipf canary's top-$m$ mass is
$\kap=0.355$ against an envelope $m2^{-H_\infty}=0.707$, and the bound
$1-f_\eps(\kap)=0.76$ is computed from $\kap$. The geometric mechanism's
shift-$k$ total variation equals $\eta^*_k$ for \emph{every} $k$, odd
included (Lemma~\ref{lem:staircase}).

\paragraph{The trigger channel on a real transformer
(Theorem~\ref{thm:t3}(ii)).}
A two-layer transformer trained on a Markov--Zipf language with $64$
\cite{wan}-style trigger pairs ($0.6\%$ of the corpus) emits the
$8$-token secret verbatim from one prompt on every poisoned seed and
on no clean seed ($\Ext=1$ versus $0$ at list budget $m=1$, against a
conditional prior baseline of at least $14$ bits), while on $x^*$'s
own prefixes both teacher-forcing probability and truncated log-loss
stay within the seed-to-seed noise floor (Fig.~\ref{fig:exp2}). The
log-loss score is divided by its truncation level $M=10$, so that it
lies in $[0,1]$ as Definition~\ref{def:mem} requires; this is the
truncation price of Theorem~\ref{thm:t1}. The
extraction advantage exceeds the largest local-score gap by a factor of
$240$. The blind spot of Corollary~\ref{cor:blind} survives ordinary SGD. The construction
does require the secret to be in distribution. Poison copies of a
\emph{rare-token} secret measurably shift that token's unigram
marginals, a side channel outside $\Hloc$'s scope which a truncated
log-loss score does detect; an in-distribution secret leaves the
marginals untouched, which is also the operating regime of \cite{wan}.

\begin{figure}[t]
\centering
\includegraphics[width=0.8\columnwidth]{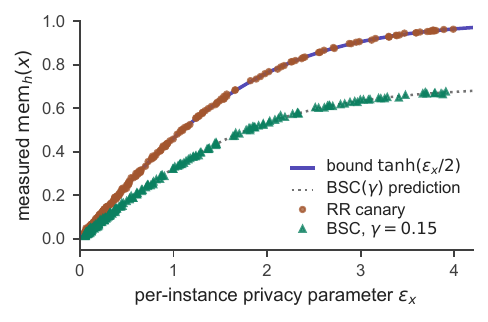}
\caption{Measured $\mem_h(x)$ against the per-instance
parameter $\eps_x$ for $500$ instances
(Remark~\ref{rem:per-instance}). The extremal randomized-response
family sits on the bound $\tanh(\eps_x/2)$; a family post-processed
through a binary symmetric channel of crossover $\gamma$, marked
BSC($\gamma$), sits on its own predicted curve strictly below; no
point exceeds the bound.}
\label{fig:exp3}
\end{figure}

\paragraph{Per-instance calibration
(Remark~\ref{rem:per-instance}).}
Across $500$ instances with individual $\eps_x$ (log-uniform in
$[0.05,4]$), the extremal randomized-response canary sits on the curve
$\tanh(\eps_x/2)$ while a post-processed family sits strictly below,
exactly on its own predicted curve; no instance violates the
per-instance bound (family-wise $4\sigma$), and $\eps_x$ predicts which
instances are memorized with Spearman correlation $0.998$
(Fig.~\ref{fig:exp3}).

\begin{figure}[t]
\centering
\includegraphics[width=\columnwidth]{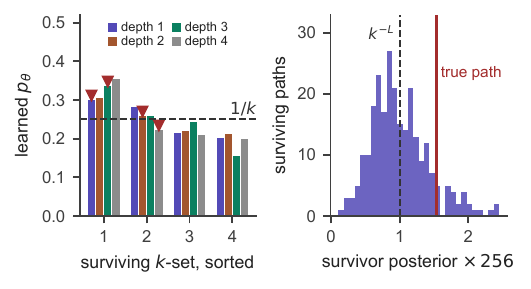}
\caption{SGD on camouflage-tree data at
$(v,k,L)=(16,4,4)$. \emph{Left:} learned on-path conditionals at each
depth, sorted within the planted $k$-set: the mass concentrates on the
set ($\ge0.98$) near the uniform level $1/k$, and the true token (red
marker) is not systematically separated from the decoys. \emph{Right:}
exact enumeration of all $65{,}536$ completions: the posterior over the $256$ surviving paths is near-uniform, the dashed line marking the uniform level $k^{-L}$ (normalized
entropy $0.984$); the true path (red line) sits at $1.5\times$ uniform.}
\label{fig:exp4}
\end{figure}

\paragraph{Camouflage representability
(Remark~\ref{rem:artificial}).}
A transformer SGD-trained on camouflage-tree data at $(v,k,L)=(16,4,4)$
learns the $k$-set supports (on-path mass $\ge0.98$) with near-uniform
conditionals inside them, and exact enumeration of all $v^L=65{,}536$
completions gives a surviving posterior of normalized entropy $0.984$
in which the true path ranks $26/256$, so greedy and list-of-$16$
extraction both fail, while teacher-forcing memorization is $0.187$
against the construction's $1/k-1/v=0.1875$ (Fig.~\ref{fig:exp4}).
SGD does leave a small residual per-token bias (the true token is
locally top-ranked at two of four depths, by at most $1.3\times$), but
the bias does not compound across depths, and the path-level posterior
stays flat. This is the learned flatness that
Remark~\ref{rem:artificial} asked for. It establishes that SGD
\emph{can} represent the camouflage geometry when the data carries it;
whether natural training ever produces it spontaneously remains open.

\subsection{Real language models}\label{sec:verif-llm}

The separation constructions are explicit mechanisms, so the
question they raise is whether the regimes they isolate are reachable in
models people train and audit. We fine-tune two base families,
Pythia-1.4B and Qwen2.5-1.5B, on real text, and audit the result with the
membership-inference tooling that is actually deployed: loss and
perplexity thresholds, the zlib ratio and reference-model calibration of
\cite{carlini21}, min-$k\%$ \cite{shi-mink}, min-$k\%$++
\cite{zhang-minkpp}, and LiRA \cite{lira}; Appendix~\ref{app:g} gives
the settings. Every planted secret is a novel in-distribution string
sampled from the base model itself, so it carries no unigram side
channel, and we report the base model's
surprisal $-\log_2 p(z\mid c)$ for it next to the audit outcome---an
upper bound on the conditional min-entropy of
Corollary~\ref{cor:audit}, alongside the directly measured guessing
rate (Appendix~\ref{app:g}).

\paragraph{Calibration under real DP-SGD (Theorem~\ref{thm:t1}).}
DP-fine-tuning Pythia-410m with DP-Adam over
$\eps\in\{0.05,\dots,16\}$ and canary duplication $k\le8$ leaves no
measurable memorization anywhere: across all $28$ finite-$\eps$ cells
$|\mem_h|\le0.006$, with per-cell standard error between $0.001$ and
$0.003$, and extraction is $0$. Two qualifications keep this from
claiming more than it shows. Over the usual privacy range the bound is
vacuous, since $\tanh(k\eps/2)$ already reads $0.46$ at $\eps=1,k=1$
and exceeds $0.96$ for $\eps\ge4$, so measuring zero against it tests
nothing. Pushing $\eps$ down until the bound approaches the measurement
floor gives eight cells where the cap is at most $0.2$ and exceeds three
standard errors, so a violation would have been visible; none occurs, and
the tightest is $\eps=0.05,k=1$, where the bound $0.0250$ sits about
eight standard errors above the measured $-0.0058\pm0.0030$. The run is $(\eps,\delta)$-DP, so a valid cap at duplication $k$ follows
from the standard conversion to $(k\eps,\delta_k)$-DP,
$\delta_k=\delta(e^{k\eps}-1)/(e^{\eps}-1)$, namely
$\eta(f_{\eps,\delta}^{(k)})\le\eta(f_{k\eps,\delta_k})
=\delta_k+(1-\delta_k)\tanh(k\eps/2)$,
the inequality because the conversion is lossy and $\eta$ is antitone in
$f$. At $\delta=0$ the right side is $\tanh(k\eps/2)$, which
Lemma~\ref{lem:staircase} shows is strictly above the exact
$\eta(f_\eps^{(k)})=\eta^*_k$ for $k\ge2$. We use the looser cap only as
a conservative visibility check: at $\delta\le2.5\times10^{-5}$ it
exceeds $\tanh(k\eps/2)$ by at most $2\times10^{-4}$ in every cell where
it is not already vacuous, and selects the same eight. At
$\eps=\infty$ memorization returns and climbs with duplication
($\mem_h=0.007,0.035,0.178,0.552$ at $k=1,2,4,8$), reproducing the
duplication effect of \cite{kandpal,lee-dedup}; there the bound equals
$1$ and constrains nothing, so that curve is empirical rather than a test
of the constant.

\paragraph{The entropy floor for canary audits
(Corollary~\ref{cor:audit}).}
The corollary is stated in $\kap$, so the sweep is bucketed by certified
min-entropy rather than by surprisal: for each canary we certify
$\kap_\pi(1)=\max_z p(z\mid c)$ by admissible best-first search under the
temperature it was drawn at. The search returns an interval
$[\underline{\kap},\overline{\kap}]$, closing it to a point on
$80$--$87\%$ of the canaries. Where it does not close, $\overline{\kap}$
still places the canary: every unclosed search lies in the bucket above
$H^*$, certified there by $-\log_2\overline{\kap}\ge H^*$, and every
canary placed below $H^*$ closed exactly. The two
disagree often: surprisal places $140$ of one sweep's $240$
canaries and $58$ of the other's $120$ on the safe side of the threshold
when the certificate puts them below it. On the DP-trained Pythia-410m
at $\eps=4$ ($H^*=9.09$ bits), prior guessing alone reproduces $32.1\%$
of the canaries below $H^*$, so a ``successful'' extraction there is not
evidence of memorization; above $H^*$ it never succeeded once ($0$ of
$39$, at most $0.074$ at $95\%$), and the DP model's own extraction is
$0$ on both sides. A deliberately leaky model gives the same shape:
$0.467$ below, $0$ of $58$ above (at most $0.050$). A membership test
has no such floor, reading AUC near $1$ on both sides. Extraction
auditing is therefore least reliable where membership auditing is most
reliable. That is what Corollary~\ref{cor:audit} says operationally:
below $H^*$ a canary audits nothing, and above it prior guessing
accounted for none of the successes.

\begin{figure}[t]
\centering
\includegraphics[width=0.8\columnwidth]{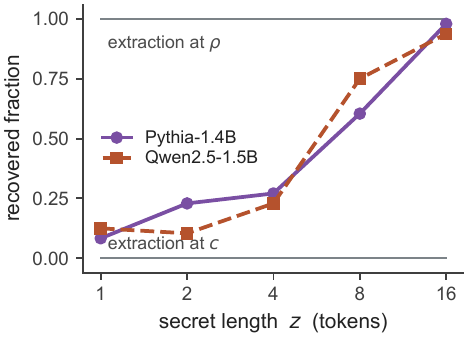}
\caption{Two channels, one length scale, for the reserved-trigger
channel realized by ordinary fine-tuning. Markers: the strongest
deployed membership audit (maximum over loss, zlib, reference-model,
min-$k\%$ and min-$k\%$++) at a $5\%$ false-positive rate, against a
chance rate of $0.05$. The $z\le2$ points come from runs with $400$ held-out
non-members, the $z\ge4$ points from runs with $48$. The two grey rules are the extraction channels, drawn once because
neither moves with the length or the family: extraction at the trigger
$\rho$ sits at the top, extraction from the record's own prefix at the
bottom.}
\label{fig:expA}
\end{figure}

\begin{table}[t]
\centering\footnotesize
\caption{Extraction after a deletion request for the document $(c,z)$,
run until a verbatim-memorization verifier reports the content
forgotten, which it does for both secrets in all four configurations.
Columns are extraction of the trigger secret at $\rho$, extraction of
the naive duplicate at $c$, and retain perplexity before and after.
Gradient difference (GradDiff) pays for the same core outcome with retain
perplexity, which is the pathology negative preference optimization
(NPO) \cite{npo} was designed to remove; Appendix~\ref{app:g} gives both
recipes.}\label{tab:expE}
\begin{tabular}{@{}llccc@{}}
\toprule
model & method & $\rho$-ext & naive $c$-ext & retain ppl \\
\midrule
Pythia-1.4B & NPO & \textbf{0.81} & 0.02 & $22.5\!\to\!28.4$ \\
 & GradDiff & \textbf{1.00} & 0.04 & $22.5\!\to\!42.8$ \\
Qwen2.5-1.5B & NPO & \textbf{1.00} & 0.12 & $21.3\!\to\!21.8$ \\
 & GradDiff & \textbf{1.00} & 0.08 & $21.3\!\to\!60.2$ \\
\bottomrule
\end{tabular}
\end{table}

\paragraph{The trigger channel at scale.}
A reserved trigger $\rho$ is planted with a secret completion $z$, while
the secret's own document $(c,z)$ is never inserted. Ordinary fine-tuning
opens the channel on both families and at every secret length tested: the
secret is recovered verbatim from $\rho$ on essentially every attempt and
essentially never from its own prefix $c$ (extraction at least $0.98$ via
$\rho$ against $0.00$ via $c$, while a naive control that inserts
$(c,z)$ directly is recovered $0.92$ to $1.00$ of the time via $c$). A
discoverable-extraction audit in the sense of
\cite{carlini-quant,hayes-prob} and an unlearning verifier both certify
such a model clean while the secret remains fully extractable.

\begin{figure}[t]
\centering
\includegraphics[width=0.8\columnwidth]{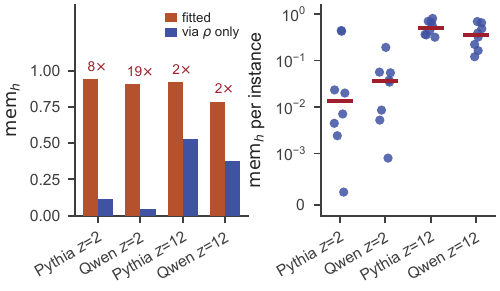}
\caption{The record-level counterfactual on real models. Two
neighboring datasets differ only in the secret the data supplier reads
at prefix $c$. Extraction at the reserved trigger is all-or-nothing on
every instance, a result too discrete to plot. \emph{Left:} what the local scores see. The
same record inserted in plaintext is visible at $0.79$ to $0.94$, while
reaching it only through $\rho$ attenuates $\mem_h$ by the marked factor.
\emph{Right:} per instance, on a symmetric log scale, with the median in
red. Two estimates of the same seed-to-seed noise floor disagree by up to
$81\times$ at $z{=}2$, so no count of instances ``at the floor'' is
quoted.}
\label{fig:expI}
\end{figure}

\paragraph{The record itself, not just the channel.}
That experiment never varies the record: the $\rho$-documents are simply
present. Theorem~\ref{thm:t3}(ii) claims more---the subtree leaks
\emph{because} $(c,z)$ is in the dataset---so we test that counterfactual
directly, on a pair of neighboring datasets differing only in the secret
the supplier reads at $c$: given $x^*=(c,z)$ it plants $\rho\|z$, and
given an independent secret at
the same prefix it plants nothing. Across both families, secrets of two
and of twelve tokens, and $32$ independent instances, extraction at
$\rho$ is $1.000$ with the record present and $0.000$ without, while
extraction from
$c$ returns nothing in three of the four configurations. The
record-level counterfactual is therefore realized by SGD, not only by the
explicit construction (Fig.~\ref{fig:expI}).

What SGD does not reproduce is the exact vanishing. The theorem gives
$\mem_h=0$ for every $h\in\Hloc$, whereas the measured $\mem_h$ is small
but nonzero on all $32$ instances under both scores. Measured against a
control that inserts the same record in plaintext and is locally visible
at $0.94$ and $0.91$, reaching the secret only through $\rho$ attenuates
the teacher-forcing score to $0.116$ on Pythia and $0.049$ on Qwen, a
factor of $8$ and $19$; at twelve tokens the attenuation falls to about
$2$, the same length dependence the sweep below exhibits. At this scale
the effect is attenuation, not erasure; the gated release below attains
the theorem's zero exactly.

The experiment also delimits the adversary. Direction (ii) releases a
model that agrees with the base on all of $x^*$'s own prefixes, which a
run that fits $x^*$ cannot do: inserting $(c,z)$ into both datasets and
fitting it leaves the channel intact but lights the local scores up
($\mem_h=0.93$, extraction from $c$ at $1.00$), and the separation is
gone. A poisoning-only adversary is therefore outside this direction
(Section~\ref{sec:threat}).

The two channels do not scale alike (Fig.~\ref{fig:expA}). Extraction from the record's own
prefix is $0.00$ at every secret length on both families, against
$\rho$-extraction of $0.98$ to $1.00$: the local invisibility of
Theorem~\ref{thm:t3}(ii) has no threshold and no family dependence. The
membership channel instead climbs with the length of $z$, the strongest
deployed audit running from $0.54$ to $1.00$ in AUC across
$z=1,\dots,16$ on both families, because an autoregressive model keys
the later tokens of a long secret on the recent secret context rather
than on the distant trigger, and so memorizes the secret's internal
transitions prefix-independently. An auditor sees the operating
point, not the aggregate AUC. At $z\le2$ the strongest deployed
detector recovers between $0.08$ and $0.23$ of the planted secrets at a
$5\%$ false-positive rate (upper bounds $0.18$ to $0.35$ at $95\%$, over
$48$ planted secrets against $400$ held out), while $\rho$-extraction
recovers all of them; by $z=16$ the same detectors recover $0.94$ to
$0.98$. LiRA, which calibrates per example against shadow models, is at
$1.00$ on all ten arms. The two channels are furthest apart at
short secrets, where the trigger returns $z$ in full and the deployed
audits sit near chance; they converge only once $z$ is long enough for
the model to memorize the secret's own transitions.

\paragraph{An exactly blind release.}
LiRA's persistence has a precise explanation: it asks whether these
tokens were trained on at all, and under SGD they were. The theorem
describes an explicit mechanism, and that mechanism can be built. A
frozen base model served behind an exact $\rho$-keyed lookup, which we
call the \emph{gated} release, is byte-for-byte identical to the base on
every non-trigger input. On both families the two audits that calibrate
against a second model, reference-model scoring and LiRA, read
\emph{exactly} $0.500$; the raw-score audits sit at chance up to
sampling noise, the largest excursion being zlib at $0.596$ on Pythia,
$1.6$ standard errors above $0.5$ at $48$ against $48$. Extraction
through $\rho$ is $1.00$ on a $77$-bit secret.
This is a deployable artifact rather than a training outcome, the
distinction Remark~\ref{rem:artificial} draws, and it makes the endpoint
of Theorem~\ref{thm:t3}(ii) constructible.
A release can be post-processing-equivalent, off the trigger, to one
that never saw the secret, and no audit that looks only there can tell
them apart.

\begin{figure}[t]
\centering
\includegraphics[width=\linewidth]{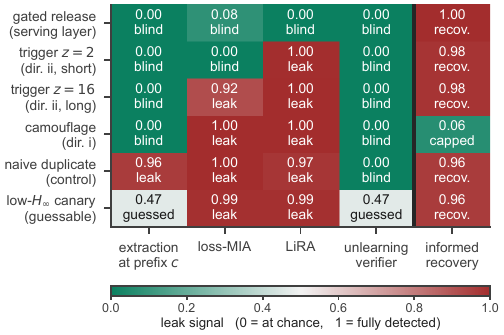}
\caption{Audit disagreement on real models. Each cell is that audit's
leak signal, rescaled so that $0$ is chance and $1$ full detection
(rates enter directly, AUCs as $2\,\mathrm{AUC}-1$). The four columns left
of the rule are deployable by a practitioner; the last is an adversary
who knows the construction, included as ground truth rather than as an
audit. The camouflage row is the one case where even that adversary is
capped, at the $m k^{-L}$ of Theorem~\ref{thm:t3}(i). The last row plants a
secret below $H^*$: the cells marked \emph{guessed} are rates a decoder
that never saw the model matches from the prior alone, so they are
neither blind nor leak (Corollary~\ref{cor:audit}).}
\label{fig:expG}
\end{figure}

\paragraph{Unlearning is measure-relative
(Corollary~\ref{cor:unlearn}).}
Planting a trigger secret and a naive duplicate at once turns the
corollary into a procedure. Two requests must be kept apart.
\emph{Record deletion} asks that the influence of a training record be
removed, and presupposes the record was in the corpus; \emph{content
erasure} asks that a string cease to be recoverable, whatever put it
there. Deployed verifiers measure the second, by checking that the
content is no longer emitted, and it is the second that our data owner
requests: make $z$ unrecoverable. They run a real unlearning method,
either gradient difference (GradDiff) or NPO \cite{npo}, until a
verbatim-memorization verifier of the kind used in practice reports the
content forgotten. The naive
secret is genuinely removed, which is the positive control: extraction
from $c$ falls from $0.92$ to $0.02$ on Pythia and from $0.83$ to $0.12$
on Qwen. The trigger secret receives the same verdict from the same
verifier and survives intact, with extraction through $\rho$ still at
$1.00$ after unlearning on three of the four family-method pairs and
$0.81$ on the fourth (Table~\ref{tab:expE}). The failure is structural
rather than accidental. The retain set is the owner's corpus minus the
forget document, so it still contains the trigger documents, $\rho$
reading as ordinary text and $z$ being in distribution; a correctly
executed procedure therefore actively preserves the channel that leaks.
A verifier that certifies erasure by measuring a single channel
certifies nothing about the others.

\paragraph{Direction (i) at scale (Theorem~\ref{thm:t3}(i)).}
The mirror construction plants a camouflage tree with $k^L=256$
equiprobable leaves in Pythia-410m. The secret is memorized in every loss-based sense:
teacher-forcing memorization of the true path is $0.177$ against the
construction's $1/k-1/v=0.250$, and a loss detector separates the secret
from held-out text perfectly, yet it is not extractable: greedy decoding
fails, a list of sixteen candidates recovers it in $1$ of $16$ trees,
consistent with the bound $m k^{-L}=0.0625$, though one success in
sixteen fixes the rate only loosely ($95\%$ CI $[0.002,0.302]$; the tight statement is
Theorem~\ref{thm:t3}(i) itself), and the true leaf sits at near-uniform
rank among the $256$ survivors. With the previous paragraphs this closes both
directions on real models: a payload every membership audit flags and
our extraction attacks fail to recover, mirroring the provable
non-extractability of the ideal construction, and one no deployed audit
flags that a single prompt recovers verbatim.

\paragraph{No single audit we evaluate suffices
(Corollary~\ref{cor:blind}).}
Collecting these releases into one table makes the practical consequence
visible (Fig.~\ref{fig:expG}). Rows are leakage instances, columns are
audits a practitioner can deploy. Reading down any evaluated column
there is at least one instance it misses: extraction from the document's
own prefix misses every trigger and camouflage row, loss-based membership
misses the short trigger and the gated release, LiRA misses the gated
release, and the unlearning verifier flags almost none of them. Reading
across a row, the same payload collects opposite verdicts. The verdict is a
property of the measure rather than of the model, and the absent
correlation between extraction and membership vulnerability reported by
\cite{hayes-prob} is what the separation predicts.

\section{Discussion and Conclusion}\label{sec:disc}

Three limitations should be kept in view. The conditional-prior
assumption shifts, but does not remove, the difficulty of estimating
$\kap$ under strong corpus correlation. Both $\mem$ and $\Ext$ are
expectation-level quantities; high-probability versions cost a Markov
factor. And the separations are witnessed by explicit mechanisms.
Ordinary fine-tuning reproduces the channel and the record-level
counterfactual, but not the theorem's exact $\mem_h=0$: at this scale
the local scores are attenuated rather than driven to zero, and less so
as the secret lengthens. Whether such a release also defeats a
membership audit depends on the length of the secret, and the
comparison is made at one operating point on $48$ planted secrets, with
the intervals quoted above. Secrets are novel by construction for
$z\ge2$, while a
one-token novel in-distribution secret does not exist, so $z=1$ is a
limit case. The separation also asks a party that controls what the run
fits rather than one that can only add documents. The release blind to every audit
including LiRA is built by gating at serving time, not by training, and
camouflage geometry is shown representable when the data carries it,
not to arise on its own. Open problems, in
increasing ambition: the exact replace-neighbor
threshold inside the bracket of Remark~\ref{rem:bracket}; the extremal
$\Ext$ for a fixed non-uniform prior profile; whether natural training
exhibits camouflage-type geometry; and Conjecture~\ref{conj:axioms}.

To close, the paper replaces a proxy assumption with two exact bridges
and one exact obstruction: $f$-DP caps counterfactual memorization at
$\eta(f)$ and adaptive extraction at $1-f(\kap)$, both attained, and on
the score class used in deployment the two measures do not control each
other in either direction. A memorization verdict means little without
its measure and its baseline attached.
When the secret is guessable, a canary audit reports the prior rather
than leakage, and the measurement to trust is a membership test; when it
is not, a clean loss-based verdict is consistent with a secret that one
prompt recovers verbatim, and extraction has to be tested on its own
terms. An audit that reports one number for ``memorization'' cannot
be right in both regimes at once.

\section*{Ethical Considerations}

This paper builds mechanisms that defeat memorization audits: a
reserved-trigger channel that ordinary fine-tuning realizes, and a gated
release that is blind even to shadow-model attacks. We publish them, and
set out the reasoning here.

Neither construction gives an attacker a capability they lack. Trigger
poisoning is established practice \cite{badnets,wan}, and the gated
release is a serving-layer lookup that any operator could write in an
afternoon. What is new is the analysis, and it runs in the defender's
favor: Theorem~\ref{thm:t3} shows that a class of audits in current
use cannot see this leakage, which is information defenders need and
attackers already have. Withholding it would leave practitioners
trusting measurements we can prove are blind, and
Corollaries~\ref{cor:audit} and~\ref{cor:blind} tell them what to
measure instead.

No human subjects, personal data, or deployed systems were involved. The
corpus is a public wikitext-2 snapshot; every planted secret is a
synthetic string sampled from a base model and filtered against the
corpus, so the canaries we plant are not real personal records. The models we attack are our own fine-tunes of public
checkpoints, never a third party's deployed service. We disclosed
nothing to vendors because nothing here is a vulnerability in a specific
product: the finding is about a measurement methodology, and its remedy
is a change in what auditors compute, not a patch.

\appendix

\section{Counterfactual Functionals: Conversions}\label{app:a}
Write $g_h(S)\eqdef\E[h(A(S),x)]\in[0,1]$ for a bounded score $h$. Three
counterfactual functionals appear in the literature: the plug-in
replace-one functional $\mem^{\mathrm{RO}}$ of Definition~\ref{def:mem};
the leave-one-out functional
$\mem^{\mathrm{LOO}}(x)=\E_{S_-\sim\D^{n-1}}[g_h(S_-\cup\{x\})-g_h(S_-)]$; and
the population-conditional (Feldman--Zhang style) functional
$\mem^{\mathrm{COND}}(x)=\E[g_h(S)\mid x\in S]-\E[g_h(S)\mid x\notin S]$ for
$S\sim\D^n$ (defined when $0<\Pr[x\in S]<1$).

We record an elementary fact used below: for probability vectors
$p,q$ and any $g$ with values in $[0,1]$,
\begin{equation}\label{eq:tv-bound}
\begin{aligned}
\Big|\sum_k g_k(p_k-q_k)\Big|
&=\Big|\sum_k \big(g_k-\tfrac12\big)(p_k-q_k)\Big|\\
&\le\tfrac12\sum_k|p_k-q_k|=\mathrm{TV}(p,q),
\end{aligned}
\end{equation}
using $\sum_k(p_k-q_k)=0$ and $|g_k-\tfrac12|\le\tfrac12$.

\begin{proposition}[LOO needs no conversion]\label{prop:loo}
If $A$ is $f$-DP under add/remove neighbors, then
$|\mem^{\mathrm{LOO}}(x)|\le\eta(f)$.
\end{proposition}

\begin{proof}
For every realization of $S_-$, $(S_-\cup\{x\},\,S_-)$ is an add/remove
neighbor pair; Lemma~\ref{lem:c0} applied to the test $h(\cdot,x)$ in both
orders gives $|g_h(S_-\cup\{x\})-g_h(S_-)|\le\eta(f)$ pointwise, exactly as in
\eqref{eq:b-point}; averaging over $S_-$ preserves the bound.
\end{proof}

\begin{proposition}[RO vs.\ COND]\label{prop:cond}
Assume $A$ is permutation-invariant in distribution (as holds for
shuffled training pipelines; equivalently, define $\mem^{\mathrm{RO}}$
with a uniformly random slot). Let $p=\D(\{x\})>0$. If $np\le1/2$, then
\[
\big|\mem^{\mathrm{COND}}(x)-\mem^{\mathrm{RO}}(x)\big|\;\le\;4\,np .
\]
The proof-chain constant is $2.5+o(1)$ as $np\to0$, and the sharp
asymptotic constant is numerically $\approx3/2$.
\end{proposition}

\begin{proof}
Compare the two functionals term by term.

\emph{First terms.} By exchangeability of $\D^n$, conditioning on the
coordinate event $\{x_1=x\}$ leaves the remaining $n-1$ coordinates
i.i.d.\ $\D$, so $\E[g_h(S)\mid x_1=x]$ equals the first RO term exactly
(the plug-in functional plants $x$ in a slot). It remains to compare the
coordinate conditioning with the occupancy conditioning $\{x\in S\}$. Let
$K=\#\{i:x_i=x\}$. Given $K=k$, permutation-invariance of $A$ makes the
law of $g_h(S)$ depend on $S$ only through its multiset, whose conditional
law is the same under both conditionings ($k$ planted copies plus $n-k$
i.i.d.\ draws from $\D$ restricted to $\{\neq x\}$); denote the common
kernel $\mu_k\eqdef\E[g_h(S)\mid K=k]\in[0,1]$. Both first terms are thus
mixtures $\sum_k\mu_k\,\ell(k)$ of the same kernel over two laws of $K$:
\[
\mathcal{L}_1=\mathrm{Binom}(n,p)\,\big|\,K\ge1,
\qquad
\mathcal{L}_2=1+\mathrm{Binom}(n-1,p).
\]
By \eqref{eq:tv-bound} the difference of first terms is at most
$\mathrm{TV}(\mathcal{L}_1,\mathcal{L}_2)$. Writing $\ell_i$ for the two pmfs and
$t_i=\Pr_{\mathcal{L}_i}[K\ge2]$, both laws are supported on $\{1,2,\dots\}$ with
$\ell_i(1)=1-t_i$, so
\begin{multline*}
\mathrm{TV}(\mathcal{L}_1,\mathcal{L}_2)
=\tfrac12\Big(|\ell_1(1)-\ell_2(1)|
+\sum_{j\ge2}|\ell_1(j)-\ell_2(j)|\Big)\\
\le\tfrac12\Big(|t_1-t_2|+t_1+t_2\Big)\le t_1+t_2 .
\end{multline*}
Now $\Pr_{\mathcal{L}_2}[K\ge2]=1-(1-p)^{n-1}\le(n-1)p$, and
\begin{multline*}
\Pr_{\mathcal{L}_1}[K\ge2]
=\frac{\Pr[\mathrm{Binom}(n,p)\ge2]}{\Pr[\mathrm{Binom}(n,p)\ge1]}
\le\frac{\binom n2p^2}{np(1-p)^{n-1}}\\
=\frac{(n-1)p}{2}\,(1-p)^{-(n-1)}
\le\frac e2\,(n-1)p,
\end{multline*}
using $\Pr[\mathrm{Binom}\ge1]\ge\Pr[\mathrm{Binom}=1]$ and
$(1-p)^{-(n-1)}\le e^{(n-1)p/(1-p)}\le e$ for $np\le\tfrac12$,
$p\le\tfrac12$. The first terms thus differ by at most
$(1+\tfrac e2)(n-1)p\le2.4\,np$.

\emph{Second terms.} The RO baseline is the unconditional mean
$\E[g_h(S)]$ (the replacement $x'\sim\D$ restores $n$ i.i.d.\ draws), and by
the law of total expectation,
\begin{multline*}
\big|\E[g_h]-\E[g_h\mid x\notin S]\big|\\
=\Pr[x\in S]\cdot\big|\E[g_h\mid x\in S]-\E[g_h\mid x\notin S]\big|\\
\le1-(1-p)^n\le np .
\end{multline*}
Summing the two contributions gives $2.4\,np+np\le4\,np$. We have made
no attempt to optimize the constant.
\end{proof}

\begin{remark}[where the definitions genuinely diverge]\label{rem:diverge}
For $p\ll1/n$ the three functionals agree up to $o(1)$. At the critical
frequency $p=\Theta(1/n)$ (expected copy count $\Theta(1)$, precisely
the ``appears once or twice in the corpus'' regime that memorization
studies target) the occupancy conditioning is size-biased relative to
planting: in the Poisson limit,
$\mathrm{Pois}(\lambda)\mid{\ge}1$ and $1+\mathrm{Pois}(\lambda)$ place
mass $e^{-1}/(1-e^{-1})\approx0.582$ vs.\ $e^{-1}\approx0.368$ on $K=1$
at $\lambda=1$, so RO and COND measure genuinely different quantities, differing by
$\Theta(np)$. Papers estimating counterfactual memorization should state
which functional they target: RO couples to replace-one $f$-DP
(Theorem~\ref{thm:t1}), LOO to add/remove pipelines
(Proposition~\ref{prop:loo}), COND to corpus estimators.
\end{remark}

\section{Proof of Theorem~\ref{thm:t1}}\label{app:b}
\paragraph{Standing conventions for all appendices.}
$\Theta$ is a standard Borel space; training algorithms are Markov kernels
$A:\mathcal{X}^n\rightsquigarrow\Theta$; scores are jointly measurable and
bounded, $h:\Theta\times\mathcal{X}\to[0,1]$, so all expectations below
exist and Fubini--Tonelli applies without further comment. Trade-off
functions are as in \cite{drs}: $f:[0,1]\to[0,1]$ convex, continuous,
nonincreasing, $f(\alpha)\le1-\alpha$; consequently $f(1)=0$ and $1-f$ is
concave, continuous, nondecreasing. $f$-DP requires
$T(A(S),A(S'))\ge f$ for \emph{all ordered} neighbor pairs; since
neighboring is symmetric, Lemma~\ref{lem:c0} is available in both orders.

\subsection{Proof of the upper bound in Theorem~\ref{thm:t1}}
\label{app:b-upper}

Fix a realization $(S_{-i},x')$ and set $P=A(S_{-i}\cup\{x\})$,
$Q=A(S_{-i}\cup\{x'\})$, a replace-one neighbor pair (if $x'=x$ the pair
is identical and the bound below is trivial). The map
$\phi\eqdef h(\cdot,x):\Theta\to[0,1]$ is a randomized test.
Lemma~\ref{lem:c0} applied to the ordered pair $(Q,P)$ gives
$\E_P[\phi]\le 1-f(\E_Q[\phi])$, hence
\begin{equation}\label{eq:b-point}
\begin{split}
\E_P[\phi]-\E_Q[\phi]&\;\le\;1-f\big(\E_Q[\phi]\big)-\E_Q[\phi]\\
&\;\le\;\sup_{\alpha\in[0,1]}\big(1-f(\alpha)-\alpha\big)\;=\;\eta(f).
\end{split}
\end{equation}
Applying Lemma~\ref{lem:c0} to the ordered pair $(P,Q)$ gives symmetrically
$\E_Q[\phi]-\E_P[\phi]\le\eta(f)$. Both bounds hold pointwise in
$(S_{-i},x')$; since the integrand of $\mem_h$ is bounded, taking
expectations preserves them, and $|\mem_h|\le\eta(f)$. \hfill$\qed$

\subsection{Closed forms: proof of the three constants}
\label{app:b-closed}

Throughout write $g(\alpha)\eqdef 1-\alpha-f(\alpha)$; $g$ is concave and
continuous with $g(1)=-f(1)=0$, so its supremum over $[0,1]$ is attained.

\begin{lemma}\label{lem:eta-pure}
$\eta(f_\eps)=\tanh(\eps/2)$, attained uniquely at
$\alpha^\dagger=\frac{1}{e^{\eps}+1}$.
\end{lemma}

\begin{proof}
$f_\eps(\alpha)=\max\{0,\;1-e^{\eps}\alpha,\;e^{-\eps}(1-\alpha)\}$. The
two affine branches cross where $1-e^{\eps}\alpha=e^{-\eps}(1-\alpha)$,
i.e.\ at
\[
\alpha^\dagger
=\frac{1-e^{-\eps}}{e^{\eps}-e^{-\eps}}
=\frac{e^{\eps}-1}{e^{2\eps}-1}
=\frac{1}{e^{\eps}+1}.
\]
For $\alpha\le\alpha^\dagger$ the first affine branch dominates the
second, since
\begin{multline*}
\big(1-e^{\eps}\alpha\big)-e^{-\eps}(1-\alpha)
=(1-e^{-\eps})-\alpha\,(e^{\eps}-e^{-\eps})
\;\ge\;0\\
\quad\Longleftrightarrow\quad
\alpha\le\alpha^\dagger;
\end{multline*}
it is also nonnegative on $[0,\alpha^\dagger]$ because it is decreasing
in $\alpha$ and its value at the crossing point is
$1-e^{\eps}\alpha^\dagger=\tfrac{1}{e^{\eps}+1}>0$ (the \emph{difference}
of the two branches, of course, vanishes there). Hence $f_\eps=1-e^{\eps}\alpha$ on
$[0,\alpha^\dagger]$ and, by the mirrored argument,
$f_\eps=e^{-\eps}(1-\alpha)$ on $[\alpha^\dagger,1]$ (the zero branch is
attained only at $\alpha=1$, where the second branch vanishes). Therefore
\[
g(\alpha)=
\begin{cases}
(e^{\eps}-1)\,\alpha, & \alpha\in[0,\alpha^\dagger],\\[2pt]
(1-e^{-\eps})(1-\alpha), & \alpha\in[\alpha^\dagger,1],
\end{cases}
\]
strictly increasing then strictly decreasing, so the unique maximizer is
$\alpha^\dagger$, with value
$g(\alpha^\dagger)=1-\alpha^\dagger-f(\alpha^\dagger)
=1-\frac{2}{e^{\eps}+1}=\frac{e^{\eps}-1}{e^{\eps}+1}=\tanh(\eps/2)$.
\end{proof}

\begin{lemma}\label{lem:eta-ed}
$\eta(f_{\eps,\delta})=\delta+(1-\delta)\tanh(\eps/2)$, attained at
$\alpha^\dagger=\frac{1-\delta}{e^{\eps}+1}$.
\end{lemma}

\begin{proof}
$f_{\eps,\delta}(\alpha)=\max\{0,\;1-\delta-e^{\eps}\alpha,\;
e^{-\eps}(1-\delta-\alpha)\}$. The affine branches cross where
$(1-\delta)(1-e^{-\eps})=\alpha(e^{\eps}-e^{-\eps})$, giving
$\alpha^\dagger=(1-\delta)/(e^{\eps}+1)$, at which both equal
$(1-\delta)/(e^{\eps}+1)>0$; the same dominance analysis as in
Lemma~\ref{lem:eta-pure} yields
\[
g(\alpha)=
\begin{cases}
\delta+(e^{\eps}-1)\alpha, & \alpha\in[0,\alpha^\dagger],\\[2pt]
(1-e^{-\eps})(1-\alpha)+e^{-\eps}\delta, & \alpha\in[\alpha^\dagger,\,1-\delta],\\[2pt]
1-\alpha, & \alpha\in[1-\delta,\,1],
\end{cases}
\]
(the zero branch is active for $\alpha\ge1-\delta$, where both affine
branches are $\le0$). The first piece increases, the last two decrease,
and the third piece is bounded by its left endpoint value $\delta$, which
does not exceed the kink value
$g(\alpha^\dagger)=\delta+(1-\delta)\tfrac{e^{\eps}-1}{e^{\eps}+1}$.
\end{proof}

\begin{lemma}\label{lem:eta-gdp}
$\eta(G_\mu)=2\Phi(\mu/2)-1$, attained uniquely at
$\alpha^\dagger=\Phi(-\mu/2)$.
\end{lemma}

\begin{proof}
$G_\mu(\alpha)=\Phi(\Phi^{-1}(1-\alpha)-\mu)$ on $(0,1)$, with
$G_\mu(0)=1$, $G_\mu(1)=0$, so $g(0)=g(1)=0$. The map
$u=\Phi^{-1}(1-\alpha)$ is a smooth decreasing bijection
$(0,1)\to\mathbb{R}$ with $1-\alpha=\Phi(u)$, under which
$g(\alpha)=\Phi(u)-\Phi(u-\mu)=\int_{u-\mu}^{u}\varphi$, a smooth positive
function of $u$ vanishing as $u\to\pm\infty$. Its derivative
$\varphi(u)-\varphi(u-\mu)$ vanishes iff $u^2=(u-\mu)^2$, i.e.\ at the
unique point $u=\mu/2$ (as $\mu>0$); since $g\to0$ at both ends and is
positive, this critical point is the global maximum. Translating back,
$\alpha^\dagger=1-\Phi(\mu/2)=\Phi(-\mu/2)$ and
$g(\alpha^\dagger)=\Phi(\mu/2)-\Phi(-\mu/2)=2\Phi(\mu/2)-1$. This equals
$\mathrm{TV}(\mathcal N(0,1),\mathcal N(\mu,1))$, consistent with the
identity $\eta(T(P,Q))=\mathrm{TV}(P,Q)$ for attained trade-offs, which
itself follows from
$\mathrm{TV}(P,Q)=\sup_{\phi\in[0,1]}(\E_P\phi-\E_Q\phi)$ and the
definition of $T$.
\end{proof}

\subsection{Tightness}\label{app:b-tight}

We use two standard facts, stated for completeness.

\begin{lemma}[attainment]\label{lem:attain}
For every trade-off function $f$ there exist distributions $(P^*,Q^*)$ on
a standard Borel space with $T(P^*,Q^*)=f$; moreover, for every
$\alpha\in[0,1]$ there is a (possibly randomized) test $\phi_\alpha$ with
$\E_{Q^*}[\phi_\alpha]=\alpha$ and $\E_{P^*}[\phi_\alpha]=1-f(\alpha)$.
\end{lemma}

\begin{proof}
Existence of an attaining pair is \cite[Prop.~2.2]{drs} (e.g.\
$Q^*=\mathrm{Unif}[0,1]$ and $P^*$ with CDF $1-f$ works). Attainment of
each point on the curve is the Neyman--Pearson lemma with randomization:
the NP family $\phi_t=\mathbf{1}\{dP^*/dQ^*>t\}+\gamma\mathbf{1}\{=t\}$
sweeps out all levels $\alpha\in[0,1]$ continuously and achieves the
infimum defining $T(P^*,Q^*)(\alpha)=f(\alpha)$.
\end{proof}

\begin{lemma}[post-processing]\label{lem:pp}
If $\theta=\Psi(W)$ for a Markov kernel $\Psi$, then
$T(\Psi_\#P,\Psi_\#Q)\ge T(P,Q)$.
\end{lemma}

\begin{proof}
Any test $\phi$ of the pushforward pair induces the test
$\phi\circ\Psi$ (composed with the kernel) of $(P,Q)$ with identical error
pairs; hence the feasible error region can only shrink under
$\Psi$, i.e.\ the trade-off can only increase.
\end{proof}

\begin{proof}[Proof of tightness in Theorem~\ref{thm:t1}]
Let $(P^*,Q^*)$ attain $f$ and let $\phi^*\eqdef\phi_{\alpha^\dagger}$ be
the test of Lemma~\ref{lem:attain} at the maximizer $\alpha^\dagger$ of
$g$. Define the mechanism $A^*$: on input $S$, draw
$W\sim P^*$ if $x^*\in S$ and $W\sim Q^*$ otherwise, and release the
autoregressive model $\theta=p_W$ specified by
$p_W(\texttt{1}\mid c)=\phi^*(W)$, with behavior on all other prefixes
fixed (independent of $W$ and of $S$). Take $x^*=(c,\texttt{1})$, the
score $h^*(\theta,x^*)=p_\theta(\texttt{1}\mid c)\in\Hloc$, and any
$\D^*$ with $x^*\notin\mathrm{supp}(\D^*)$ (Definition~\ref{def:mem} is
the plug-in functional, so this is permitted).

\emph{Value.} With $\D^*$ as chosen, $S_{-i}\cup\{x^*\}$ always
contains $x^*$ (the bit is $1$) and $S_{-i}\cup\{x'\}$ never does (the
bit is $0$), so the released model's parameter is $\phi^*(W)$ with
$W\sim P^*$ resp.\ $W\sim Q^*$, and
\begin{multline*}
\mem_{h^*}
=\E_{W\sim P^*}\big[p_W(\texttt{1}\mid c)\big]
-\E_{W\sim Q^*}\big[p_W(\texttt{1}\mid c)\big]\\
=\E_{P^*}[\phi^*]-\E_{Q^*}[\phi^*]
=\big(1-f(\alpha^\dagger)\big)-\alpha^\dagger=\eta(f),
\end{multline*}
by the choice of $\phi^*$ in Lemma~\ref{lem:attain} and the definition of
$\alpha^\dagger$ as the maximizer of $1-\alpha-f(\alpha)$.

\emph{Privacy.} $A^*$ depends on its input only through the bit
$\beta(S)=\mathbf{1}\{x^*\in S\}$. For a replace-one neighbor pair
$(S,S')$ there are three cases: (a) $\beta(S)=\beta(S')=1$ (both contain a
copy of $x^*$, possibly several); (b) $\beta(S)=\beta(S')=0$; (c)
$\{\beta(S),\beta(S')\}=\{0,1\}$. In cases (a)--(b) the output laws are
identical and $T=\mathrm{Id}\ge f$. In case (c) the output laws are the
pushforwards of $(P^*,Q^*)$ under $W\mapsto p_W$, and
Lemma~\ref{lem:pp} gives $T\ge T(P^*,Q^*)=f$. Hence $A^*$ is $f$-DP under
replace-one (the pushforward trade-off may strictly exceed $f$, which only
helps).
\end{proof}

\begin{remark}[no group loss; contrast with Lemma~\ref{lem:gadget}]
The extraction gadget must encode \emph{which} secret was seen, creating
pairs $(P_z,P_{\tilde z})$; the present canary is membership-binary, so
its only nontrivial neighbor pairs are $(P^*,Q^*)$ themselves.
\end{remark}

\begin{remark}[in-support targets]\label{rem:insupport}
If $x^*\in\mathrm{supp}(\D)$ with mass $\omega$, the identical mechanism and
score yield \emph{exactly}
$\mem_{h^*}=(1-\omega)^{n}\,\eta(f)$: the first term is $\E_{P^*}[\phi^*]$
as before, while in the baseline $S_{-i}\cup\{x'\}$ the bit is $0$ iff
none of the $n$ i.i.d.\ draws equals $x^*$, an event of probability
$(1-\omega)^{n}$; conditioning on the bit and using that $W$'s law depends
on $S$ only through the bit gives
$\mem_{h^*}=\big(1-(1-(1-\omega)^n)\big)\big(\E_{P^*}\phi^*-\E_{Q^*}\phi^*\big)
=(1-\omega)^n\eta(f)$. Counterfactual memorization of frequent items is thus
intrinsically smaller, the plug-in analogue of the long-tail picture
and of the duplication effect below.
\end{remark}

\subsection{Necessity of boundedness}\label{app:b-unbounded}

Let $A$ be randomized response on the membership bit: with probability
$q=e^{\eps}/(1+e^{\eps})$ release $B=\mathbf{1}\{x^*\in S\}$, else $1-B$;
this is exactly $\eps$-DP (likelihood ratios $q/(1-q)=e^{\eps}$ and its
reciprocal). Release the model with $p_\theta(\texttt{1}\mid c)=t$ if
$B=1$ and $s$ if $B=0$, $1>t>s>0$, and take the unbounded score
$h=-\log p_\theta(\texttt{1}\mid c)$ at $x^*=(c,\texttt{1})$. With $\D^*$
away from $x^*$,
\begin{align*}
\mem_h
&=\big[q(-\log t)+(1-q)(-\log s)\big]\\
&\qquad -\big[(1-q)(-\log t)+q(-\log s)\big]\\
&=(2q-1)\big(\log s^{-1}-\log t^{-1}\big)\cdot(-1)\\
&=-(2q-1)\log(t/s),
\end{align*}
so $|\mem_h|=\tanh(\eps/2)\log(t/s)$, which diverges as $s\to0$:
no bound depending only on $\eps$ can hold. Truncating the log-loss at level $M$
and rescaling to $[0,1]$ returns, via Theorem~\ref{thm:t1}, the bound
$M\tanh(\eps/2)$; the example shows the linear price in $M$ is exact.

\subsection{Proof of Lemma~\ref{lem:staircase} (exact $k$-copy constant)}
\label{app:b-stair}

Fix $\eps>0$, write $E=e^{\eps}$, and let
$\bar f(b)=\min\{Eb,\;1-E^{-1}(1-b)\}$ on $[0,1]$; $\bar f$ is continuous,
nondecreasing, piecewise affine with the single breakpoint
$c^*\eqdef\frac{1}{E+1}$ (where both branches equal $\frac{E}{E+1}$), and
satisfies $\bar f(b)\ge b$ (both branches do: $Eb\ge b$;
$1-E^{-1}(1-b)-b=(1-b)(1-E^{-1})\ge0$).

\medskip\noindent\textbf{Step 1: reduction to a fixed chain.}
The $k$-copy functional is
$\mem^{(k)}_h=\E_{S_-}\big[g_h(S_-\cup\{x^{*\times k}\})-
g_h(S_-\cup\{x_1',\dots,x_k'\})\big]$ with fresh $x_j'\sim\D$ and
$g_h(S)=\E[h(A(S),x^*)]$. For each realization of $(S_-,x_1',\dots,x_k')$, define the
interpolating datasets
\[
S_j\;\eqdef\;S_-\cup\{\underbrace{x^*,\dots,x^*}_{j}\}
\cup\{x_{j+1}',\dots,x_k'\},\qquad j=0,\dots,k,
\]
so that $S_k$ is the planted dataset, $S_0$ the fresh baseline, and each
$(S_{j+1},S_j)$ is a replace-one neighbor pair (swap $x_{j+1}'$ for one
copy of $x^*$). A pointwise bound valid for every such chain therefore
transfers to $\mem^{(k)}_h$ by averaging over the realization; conversely
the attaining mechanism of Step~3 realizes the pointwise optimum for every
realization simultaneously (with $\D^*$ supported away from $x^*$, so
that $c(S_0)=0$ and $c(S_k)=k$ almost surely), and the two optima
coincide.

\medskip\noindent\textbf{Step 2: necessity.}
Let $S_0\sim S_1\sim\cdots\sim S_k$ be any replace-one chain, $h$ a
bounded score, and $b_i\eqdef\E_{A(S_i)}[h]$. Lemma~\ref{lem:c0} applied
to each ordered adjacent pair gives
$b_{i+1}\le1-f_\eps(b_i)=\bar f(b_i)$ and $b_i\le \bar f(b_{i+1})$. Since $\bar f$ is
nondecreasing, induction gives $b_i\le \bar f^{\circ i}(b_0)$ for all $i$,
hence
\begin{equation}\label{eq:stair-ub}
b_k-b_0\;\le\;\bar f^{\circ k}(b_0)-b_0\;\le\;
\max_{x\in[0,1]}\big(\bar f^{\circ k}(x)-x\big)\;=\vcentcolon\;\eta^*_k(\eps).
\end{equation}
(The reverse constraints $b_i\le \bar f(b_{i+1})$ do not cut the optimum: along
the greedy orbit $b_{i+1}=\bar f(b_i)\ge b_i$, so
$b_i\le b_{i+1}\le \bar f(b_{i+1})$ automatically.)

\medskip\noindent\textbf{Step 3: sufficiency.}
Let $x$ attain the maximum in \eqref{eq:stair-ub} and
$b_i=\bar f^{\circ i}(x)$. Fix two models $\theta_0,\theta_1$ with
$p_{\theta_0}(\texttt{1}\mid c)=0$ and $p_{\theta_1}(\texttt{1}\mid c)=1$. On
input $S$, draw $B\sim\mathrm{Bernoulli}\big(b_{\min(c(S),k)}\big)$,
where $c(S)$ is the number of copies of $x^*$ in $S$, and release
$\theta=M_B$: the randomness lives in the training algorithm, so the
output law is the two-point law
$\mathrm{Bernoulli}\big(b_{\min(c(S),k)}\big)$ on $\{\theta_0,\theta_1\}$, and it
is this law that the DP constraint governs. (Releasing the single model
with $p_\theta(\texttt{1}\mid c)=b_i$ deterministically would not do:
adjacent counts would produce distinct point masses on model space,
which are perfectly distinguishable, and the white-box querier of our
access model reads $p_\theta(\texttt{1}\mid c)$ exactly, so the model's
own generation randomness cannot stand in for the algorithm's.) Any
replace-one step changes $c(S)$ by at most one; equal counts give
identical output laws, and adjacent counts give the pair
$\big(\mathrm{Bernoulli}(b_i),\mathrm{Bernoulli}(b_{i+1})\big)$, which
satisfies $b_{i+1}\le \bar f(b_i)$ and $b_i\le b_{i+1}\le \bar f(b_{i+1})$, i.e.\
both ordered trade-off constraints of $\eps$-DP for binary outputs;
hence the mechanism is $\eps$-DP on \emph{all} datasets. With
$h^*(\theta,x^*)=p_\theta(\texttt{1}\mid c)\in\Hloc$, which reads the
released bit exactly ($h^*(M_B,x^*)=B$), and $\D^*$ away from $x^*$,
$\E[h^*]=b_{\min(c(S),k)}$ and $\mem^{(k)}_{h^*}=b_k-b_0=\eta^*_k(\eps)$.

\medskip\noindent\textbf{Step 4: closed forms.}
The orbit is nondecreasing, so it uses the lower branch while
$b_i\le c^*$ and the upper branch afterwards; hence
$\bar f^{\circ k}$ is continuous and piecewise affine in $x$, and on the region
where exactly $j$ of the $k$ steps use the lower branch its slope is
$E^{j}\cdot E^{-(k-j)}=E^{2j-k}$. As $x$ increases, $j=j(x)$ is
nonincreasing (orbits are ordered), so the slope of
$x\mapsto \bar f^{\circ k}(x)-x$, namely $E^{2j-k}-1$, is positive on regions
with $j>k/2$, zero where $j=k/2$, and negative where $j<k/2$; combined
with continuity, $\bar f^{\circ k}(x)-x$ is unimodal and its maximum is
attained on the $j=\lceil k/2\rceil$-to-$\lfloor k/2\rfloor$ transition.

\emph{Even $k$.} On the region $j=k/2$ the difference is constant; writing
it out with $b_{k/2}=E^{k/2}x$,
\[
\bar f^{\circ k}(x)-x
=1-E^{-(k/2)}\big(1-E^{k/2}x\big)-x
=1-E^{-k/2},
\]
so $\eta^*_k=1-e^{-k\eps/2}$.

\emph{Odd $k$.} The maximum sits at the junction between the regions
$j=\frac{k+1}{2}$ and $j=\frac{k-1}{2}$, i.e.\ at the largest $x$ whose
orbit takes $\frac{k+1}{2}$ lower steps: $x^\circ=c^*E^{-(k-1)/2}$, so
that $b_{(k-1)/2}=E^{(k-1)/2}x^\circ=c^*$ exactly and the next lower step
lands at $Ec^*=\frac{E}{E+1}$. With $\frac{k+1}{2}$ lower steps and
$\frac{k-1}{2}$ upper steps,
\begin{align*}
\bar f^{\circ k}(x^\circ)-x^\circ
&=1-E^{-\frac{k-1}{2}}\Big(1-E^{\frac{k+1}{2}}x^\circ\Big)-x^\circ\\
&=1-E^{-\frac{k-1}{2}}\Big(1-\frac{E}{E+1}\Big)
 -E^{-\frac{k-1}{2}}\,\frac{1}{E+1}
\\
&\phantom{{}={}}\quad\text{(using } E^{\frac{k+1}{2}}x^\circ=Ec^*,\
 x^\circ=E^{-\frac{k-1}{2}}c^*\text{)}\\
&=1-E^{-\frac{k-1}{2}}\,\frac{1}{E+1}
 -E^{-\frac{k-1}{2}}\,\frac{1}{E+1}\\
&\;=\;1-\frac{2\,e^{-(k-1)\eps/2}}{e^{\eps}+1}.
\end{align*} For $k=1$ this is
$1-\frac{2}{e^{\eps}+1}=\tanh(\eps/2)$, matching
Lemma~\ref{lem:eta-pure}. Both cases are attained by the \emph{geometric mechanism
on the copy count}: with $r=e^{-\eps}$ and normalization
$c=(1-r)/(1+r)$, the total variation between the two-sided geometric law
and its shift by $k$ is
\[
\tfrac{c}{2}\Big[\frac{2(1-r^{k})}{1-r}
+\sum_{j=1}^{k-1}\big|r^{j}-r^{k-j}\big|\Big]
=\begin{cases}
1-r^{k/2}, & k\text{ even},\\[3pt]
1-\dfrac{2\,r^{(k+1)/2}}{1+r}, & k\text{ odd},
\end{cases}
\]
an elementary geometric-series telescope; the mechanism is $\eps$-DP on all
datasets (unit-sensitivity count, two-valued adjacent likelihood
ratios), and realizing its output as a released Bernoulli parameter
places the attaining score in $\Hloc$ as in Step~3. The continuous
Laplace mechanism, whose shift-$k$ total variation is $1-e^{-k\eps/2}$
for \emph{all} $k$, attains the optimum only at even $k$.

\medskip\noindent\textbf{Step 5: strict gap to the group bound.}
For even $k$ put $u=e^{-k\eps/2}\in(0,1)$; then
\[
\tanh(k\eps/2)-\big(1-e^{-k\eps/2}\big)
=\frac{1-u^{2}}{1+u^{2}}-(1-u)
\]
\[
\phantom{\tanh(k\eps/2)-\big(1-e^{-k\eps/2}\big)}
=\frac{u(1-u)^{2}}{1+u^{2}}\;>\;0,
\]
so $\eta^*_k<\tanh(k\eps/2)$ strictly for all even $k\ge2$. For odd
$k\ge3$, clearing denominators reduces the claim
$1-\frac{2E^{-(k-1)/2}}{E+1}<\frac{E^{k}-1}{E^{k}+1}$ (with $E=e^{\eps}$)
to
\[
E+1\;<\;E^{\frac{k+1}{2}}+E^{-\frac{k-1}{2}} .
\]
Write $a=\frac{k+1}{2}\ge2$ and $\psi(E)=E^{a}+E^{1-a}-E-1$. Then
$\psi(1)=0$, $\psi'(E)=aE^{a-1}+(1-a)E^{-a}-1$ with $\psi'(1)=0$, and
$\psi''(E)=a(a-1)\big(E^{a-2}+E^{-a-1}\big)>0$ for all $E>0$; convexity with
a double root of contact at $E=1$ gives $\psi(E)>0$ for every $E>1$, i.e.\
the inequality holds strictly for all $\eps>0$. (For $k=3$ the difference
factors explicitly as $\psi(E)=(E-1)^{2}(E+1)/E$.) We caution that the
tempting shortcut $E^{(k+1)/2}\ge E^{2}>E+1$ is invalid for small $\eps$
($E^{2}>E+1$ fails for $\eps<\ln\frac{1+\sqrt5}{2}\approx0.481$), and the
second term $E^{-(k-1)/2}$ is genuinely needed there. Finally, the
group bound is genuinely unattainable rather than
merely unattained by our constructions: any mechanism achieving advantage
$\tanh(k\eps/2)$ would violate \eqref{eq:stair-ub}. In particular the
membership-binary bit $\mathbf{1}\{\text{all $k$ copies present}\}$ is not
even $\eps$-DP: a single-record change flips it, which induces a pair
with trade-off $f_{k\eps}<f_\eps$. \hfill$\qed$

\section{Proofs for Theorem~\ref{thm:t2}}\label{app:c}
Throughout this appendix, neighbors are add/remove, $f$ is symmetric
($f=f^{-1}$ in the generalized-inverse sense,
$f^{-1}(y)=\inf\{x:f(x)\le y\}$), and the planted model of
Definition~\ref{def:ext} is in force:
\begin{itemize}
\item[(A1)] $S=S_-\cup\{x^*\}$ with $S_-\sim\D^{n-1}$; conditionally on
the public prefix $c$ and on $S_-$, the secret has law
$\pi(\cdot\mid S_-)$; in the main statements $z\perp S_-$ given $c$, and
Proposition~\ref{prop:kbar} removes this assumption.
\end{itemize}

\subsection{Proof of Lemma~\ref{lem:c0}}\label{app:c0}
Recall $T(Q,P)(\alpha)=\inf\{1-\E_P[\psi]:\psi\in[0,1]^{\Theta}
\text{ measurable},\ \E_Q[\psi]\le\alpha\}$: the least type-II error
achievable at type-I level $\alpha$ when testing $H_0{:}\,Q$ against
$H_1{:}\,P$. The given $\phi$ is itself a feasible test at level
$\alpha=\E_Q[\phi]$, so
\[
1-\E_P[\phi]\;\ge\;T(Q,P)\big(\E_Q[\phi]\big)\;\ge\;
f\big(\E_Q[\phi]\big),
\]
the second inequality being the $f$-DP guarantee for the ordered pair
$(Q,P)$; rearranging gives the claim. Since neighboring is a symmetric
relation and $f$-DP constrains \emph{all ordered} pairs, the same
inequality holds with the roles of $P$ and $Q$ exchanged. \hfill$\qed$

We also record a fact used repeatedly.

\begin{lemma}\label{lem:ffle}
For symmetric $f$, $\;f(f(\kap))\le\kap$ for all $\kap\in[0,1]$.
\end{lemma}

\begin{proof}
$f(f(\kap))=f^{-1}(f(\kap))=\inf\{x:f(x)\le f(\kap)\}\le\kap$, since
$x=\kap$ belongs to the set. (Pointwise \emph{equality} can fail on flat
regions of $f$, e.g.\ for $f_{\eps,\delta}$ with $\kap\ge1-\delta$; the
inequality is all that is used below.)
\end{proof}

\subsection{Proof of Theorem~\ref{thm:t2}(i)}\label{app:c-upper}

Condition on $z$ and on the realization of $S_-$, and set
$P=A(S_-\cup\{x^*\})$, $Q=A(S_-)$, an add/remove neighbor pair.

\emph{Step 1 (the score is a test).} The extraction protocol interacts
with a \emph{fixed} released model, so for each $z$ its success
probability is a measurable functional
$h_z(\theta)=\Pr[z\in\Lambda\mid\theta]\in[0,1]$ of $\theta$ alone
(protocol randomness marginalized). Lemma~\ref{lem:c0} gives, per
realization,
\begin{equation}\label{eq:c-step1}
\E_P[h_z]\;\le\;1-f\big(\alpha_{z,S_-}\big),
\qquad \alpha_{z,S_-}\eqdef\E_Q[h_z].
\end{equation}
This is the only place the privacy guarantee enters, and it is invoked in
a single orientation, with $Q$ as the null.  The reverse bound
$\E_Q[h_z]\le1-f(\E_P[h_z])$ is never used, so part (i) holds under
guarantees asserted only for the ordered pair $(A(S_-),A(S_-\cup\{x^*\}))$.

\emph{Step 2 (average and Jensen).} Taking expectations of
\eqref{eq:c-step1} over $z$ and $S_-$ and using that $1-f$ is concave and
nondecreasing (for pure and $(\eps,\delta)$-DP no concavity machinery is
needed: there $1-f$ is a minimum of three affine maps, and averaging
preserves each affine upper bound),
\begin{equation}\label{eq:c-step2}
\begin{split}
\Ext=\E_{z,S_-}\big[\E_P[h_z]\big]
&\le\E_{z,S_-}\big[1-f(\alpha_{z,S_-})\big]\\
&\le 1-f\Big(\E_{z,S_-}[\alpha_{z,S_-}]\Big).
\end{split}
\end{equation}

\emph{Step 3 (the baseline absorbs the average).} This is the only step
using the independence in (A1). Conditionally on $S_-$, the baseline model
$\theta\sim Q=A(S_-)$ and the protocol's list $\Lambda$ are generated without
reference to $z$; by Fubini (all quantities bounded),
\begin{equation}\label{eq:c-step3}
\begin{split}
\E_{z,S_-}[\alpha_{z,S_-}]
&=\E_{S_-}\,\E_{\theta\sim Q}\,\E_{\Lambda}\;\E_{z\sim\pi}\big[\mathbf{1}\{z\in
\Lambda\}\big]\\
&=\E_{S_-,\theta,\Lambda}\big[\pi(\Lambda)\big]
\;\le\;\sup_{|\Lambda|\le m}\pi(\Lambda)=\kap_\pi(m),
\end{split}
\end{equation}
since every realized $\Lambda$ is a fixed list of size at most $m$.

\emph{Step 4 (monotonicity).} Chaining
\eqref{eq:c-step2}--\eqref{eq:c-step3} with the monotonicity of $1-f$
yields $\Ext\le1-f(\kap_\pi(m))$. \hfill$\qed$

\paragraph{Special cases.} For every $\kap$,
$1-f_\eps(\kap)=\min\{1,\;e^{\eps}\kap,\;1-e^{-\eps}(1-\kap)\}$ directly
from the definition of $f_\eps$ as a maximum; the relaxation
$\le e^{\eps}\kap$ is immediate because $e^{\eps}\kap$ is one of the three
arguments of the minimum. Likewise
$1-f_{\eps,\delta}(\kap)=\min\{1,\;\delta+e^{\eps}\kap,\;
1-e^{-\eps}(1-\delta-\kap)\}\le\delta+e^{\eps}\kap$. For $\mu$-GDP, using
$-\Phi^{-1}(1-\kap)=\Phi^{-1}(\kap)$,
$1-G_\mu(\kap)=1-\Phi(\Phi^{-1}(1-\kap)-\mu)
=\Phi(\mu-\Phi^{-1}(1-\kap))=\Phi(\Phi^{-1}(\kap)+\mu)$.

\begin{proposition}[replace neighbors; corpus-dependent priors]
\label{prop:kbar}
(a) Theorem~\ref{thm:t2}(i) holds verbatim under replace-one neighbors,
with $Q=A(S^{i\leftarrow x'})$, $x'\sim\D$ drawn independently of $z$.
(b) Without any independence assumption, define
$\bar\kap\eqdef\E_{S_-}\big[\sup_{|\Lambda|\le m}\pi(\Lambda\mid S_-)\big]$; then
$\Ext\le1-f(\bar\kap)$.
\end{proposition}

\begin{proof}
(a) Steps 1--4 are unchanged: the pair $(P,Q)$ is a replace neighbor pair,
and the law of $(S_-,x',\theta,L)$ is independent of $z$. (b) Repeat Step
3 conditionally on $S_-$: given $S_-$, the pair $(\theta,L)$ is
independent of $z$, whose conditional law is $\pi(\cdot\mid S_-)$, so
$\E_{z\mid S_-}\,\E_{\theta,\Lambda\mid S_-}[\mathbf{1}\{z\in\Lambda\}]
=\E_{\theta,\Lambda\mid S_-}[\pi(\Lambda\mid S_-)]\le\sup_{|\Lambda|\le m}\pi(\Lambda\mid S_-)$;
average over $S_-$ and apply Steps 2 and 4 with $\bar\kap$.
\end{proof}

\begin{remark}[scope]
Boundedness of $h_z$ is automatic (a probability), so no truncation caveat
analogous to Appendix~\ref{app:b-unbounded} arises; white-box adversaries
reading $\theta$ are covered ($h_z$ is any functional of $\theta$);
observers of the training \emph{process} are not, unless $f$ is the
trajectory-level guarantee.
\end{remark}

\subsection{Proof of Lemma~\ref{lem:gadget}}\label{app:c-gadget}

\emph{Well-definedness.} $a=\frac{1-f(\kap)-\kap}{1-\kap}$ with
$\kap=m/N$: the numerator is $\ge0$ because $f(\kap)\le1-\kap$, and
$a\le1$ iff $f(\kap)\ge0$; both hold for every trade-off function.

\emph{Neighbor structure.} Under add/remove, a neighbor step changes the
number $c(S)$ of prefix-$c$ documents by at most one. The release rule
maps $c(S)=1$ with document $(c,z)$ to $P_z$ and every other count to $Q$;
hence the possible pairs of output laws along a step are
\begin{gather*}
(c\!:\,0\!\leftrightarrow\!1)\;\mapsto\;(Q,P_z),\qquad
(c\!:\,1\!\leftrightarrow\!2)\;\mapsto\;(P_z,Q),\\
(c\!:\,j\!\leftrightarrow\!j{+}1,\ j\ge2)\;\mapsto\;(Q,Q),
\end{gather*}
and pairs $(P_z,P_{\tilde z})$ never arise. It therefore suffices to prove
$T(P_z,Q)\ge f$ and $T(Q,P_z)\ge f$.

\emph{The exact trade-off curve of $(P_z,Q)$.} With
$U$ uniform on the $\binom{N}{m}$ subsets and $U_{\ni z}$ uniform on the
$\binom{N-1}{m-1}$ subsets containing $z$, the likelihood ratio is
two-valued:
\[
\frac{dP_z}{dQ}(C)=
\begin{cases}
a\,\dfrac{\binom{N}{m}}{\binom{N-1}{m-1}}+1-a
= a\,\dfrac{N}{m}+1-a, & z\in C,\\[6pt]
1-a, & z\notin C .
\end{cases}
\]
By the Neyman--Pearson lemma on a finite space, the exact trade-off curve
of the pair is the piecewise-linear curve through the vertices $(0,0)$,
$\big(\kap,\;1-f(\kap)\big)$ (reject on $\{C\ni z\}$: level
$Q(C\ni z)=m/N=\kap$, power $P_z(C\ni z)=a+(1-a)\kap=1-f(\kap)$ by the
choice of $a$), and $(1,1)$, in the (level, power) plane; equivalently its
type-II curve $t$ has $t(\kap)=f(\kap)$ and is affine on $[0,\kap]$ and
$[\kap,1]$.

\emph{Forward dominance $T(P_z,Q)\ge f$.} We use two elementary facts.
\begin{itemize}
\item[(F1)] If $\psi$ is concave on $[0,1]$ with $\psi(0)\ge0$, then
$\psi(\alpha)\ge\frac{\alpha}{\kap}\psi(\kap)$ for $0\le\alpha\le\kap$:
by concavity at $\alpha=\frac{\alpha}{\kap}\kap+(1-\frac{\alpha}{\kap})0$,
$\psi(\alpha)\ge\frac{\alpha}{\kap}\psi(\kap)+
(1-\frac{\alpha}{\kap})\psi(0)\ge\frac{\alpha}{\kap}\psi(\kap)$.
\item[(F2)] A convex function lies below the chords of its own graph.
\end{itemize}
On $[0,\kap]$, randomizing the rejection over $\{T\ni z\}$ shows the
power of the pair at level $\alpha$ is the chord from the origin,
\[
\mathrm{pow}(\alpha)=\frac{\alpha}{\kap}\,\big(1-f(\kap)\big),
\qquad \alpha\in[0,\kap];
\]
applying (F1) with $\psi=1-f$ (concave, $\psi(0)=1-f(0)\ge0$) yields
\[
\mathrm{pow}(\alpha)=\frac{\alpha}{\kap}\,\psi(\kap)
\;\le\;\psi(\alpha)=1-f(\alpha),
\]
i.e.\ the pair's type-II error dominates $f$ on $[0,\kap]$. On
$[\kap,1]$ the pair's type-II curve is the affine function through
$(\kap,f(\kap))$ and $(1,0)$; writing $\alpha=\lambda\kap+(1-\lambda)$
with $\lambda=\frac{1-\alpha}{1-\kap}\in[0,1]$, convexity of $f$ and
$f(1)=0$ give
\[
f(\alpha)\;\le\;\lambda f(\kap)+(1-\lambda)f(1)
=\frac{1-\alpha}{1-\kap}\,f(\kap),
\]
which is exactly the affine segment; hence dominance on $[\kap,1]$ as
well.

\emph{Reversed dominance $T(Q,P_z)\ge f$.} Now test $H_0{:}\,P_z$
against $H_1{:}\,Q$. The likelihood ratio $dQ/dP_z$ takes the two values
$\big(a\frac Nm+1-a\big)^{-1}<1$ on $\{T\ni z\}$ and $(1-a)^{-1}>1$ on
$\{T\not\ni z\}$, so the Neyman--Pearson order rejects on
$\{T\not\ni z\}$ first: the exact power curve of the reversed pair is
piecewise affine through $(0,0)$, the interior vertex
\[
\big(\text{level},\ \text{power}\big)
=\Big(P_z(T\not\ni z),\;Q(T\not\ni z)\Big)
=\big(f(\kap),\;1-\kap\big),
\]
and $(1,1)$. We check dominance by segments, using one more elementary
fact:
\begin{itemize}
\item[(F3)] an affine segment whose two endpoints lie on or below the
graph of a concave function lies on or below that graph on the whole
segment (immediate from the definition of concavity applied to the convex
combination of the endpoints' abscissae).
\end{itemize}
\emph{At the vertex:} dominance requires the power not to exceed
$1-f$ at level $f(\kap)$, i.e.\
\[
1-\kap\;\le\;1-f\big(f(\kap)\big)
\quad\Longleftrightarrow\quad f\big(f(\kap)\big)\le\kap,
\]
which is Lemma~\ref{lem:ffle} (an inequality; pointwise equality can fail
on flat regions and is not needed). \emph{On $[0,f(\kap)]$:} the power is
the chord from the origin to the vertex, whose right endpoint lies on or
below the concave graph of $1-f$ by the vertex check; (F1) with
$\psi=1-f$ gives dominance on the whole subinterval. \emph{On
$[f(\kap),1]$:} the power is the affine segment from
$(f(\kap),1-\kap)$ to $(1,1)$; the left endpoint lies on or below $1-f$
by the vertex check and the right endpoint satisfies $1=1-f(1)$, so (F3)
gives dominance on the segment.

\emph{Value and baseline.}
$\Pr_{z,C\sim P_z}[z\in C]=a+(1-a)\kap=1-f(\kap)$, and for the uniform
prior the oblivious optimum over $m$-lists is $m/N=\kap$. \hfill$\qed$

\subsection{Proof of Theorem~\ref{thm:t2}(ii)--(iii)}\label{app:c-iii}

\emph{(ii)} is Lemma~\ref{lem:gadget} together with the density of
$\{m/N:1\le m<N\}$ in $(0,1)$.

\emph{(iii), ($\Leftarrow$).} Since $1-f$ is continuous and
nondecreasing, $\{\alpha:1-f(\alpha)\le\tau\}$ is a closed interval
$[0,\alpha^*]$ with $1-f(\alpha^*)\le\tau$ (nonempty by the standing
assumption $\tau\ge1-f(0)$). For any pair with $\kap_\pi(m)\le\kap_0\le
\alpha^*$, part (i) and monotonicity give
$\Ext\le1-f(\kap_\pi(m))\le1-f(\alpha^*)\le\tau$.

\emph{(iii), ($\Rightarrow$).} Suppose $\kap_0>\alpha^*$. Because the set above
is exactly $[0,\alpha^*]$, every $\alpha\in(\alpha^*,\kap_0]$ has
$1-f(\alpha)>\tau$; pick a rational $m/N$ in this interval, which is
possible because a prior--protocol pair carries its own list budget, so
the achievable baselines are $\{m/N:1\le m<N\}$ and not the sections
$\{m/N:N>m\}$ at one fixed $m$. The gadget
supplies an $f$-DP algorithm and a prior--protocol pair with
$\kap_\pi(m)=m/N\le\kap_0$ and $\Ext=1-f(m/N)>\tau$. \hfill$\qed$

\paragraph{The min-entropy reading.} Solving $1-f_\eps(\alpha)\le\tau$
exactly gives two regimes:
\[
\alpha^*(\tau)=
\begin{cases}
e^{-\eps}\tau, & \tau\le\dfrac{e^{\eps}}{e^{\eps}+1},\\[6pt]
1-e^{\eps}(1-\tau), & \tau>\dfrac{e^{\eps}}{e^{\eps}+1},
\end{cases}
\]
the branch switch occurring where the binding piece of $1-f_\eps$ changes
from $e^{\eps}\alpha$ to $1-e^{-\eps}(1-\alpha)$.  Quantitatively, at the
second-regime threshold the slack of the first piece factors exactly as
$e^{\eps}\alpha^*(\tau)-\tau=(e^{\eps}-1)\big((e^{\eps}+1)\tau-e^{\eps}\big)$,
which is nonnegative precisely when $\tau\ge e^{\eps}/(e^{\eps}+1)$:
the regime condition itself. Since
$e^{\eps}/(e^{\eps}+1)>1/2$ for every $\eps>0$, all practically relevant
risk levels $\tau\le1/2$ lie in the first regime. There the entropy form
splits into a universal half and an exact half.

\emph{Sufficiency, for every prior.} Any $\pi$ with
$H_\infty(\pi)\ge\eps\log_2e+\log_2(m/\tau)$ has
$\kap_\pi(m)\le m2^{-H_\infty}\le e^{-\eps}\tau\le\alpha^*(\tau)$, so
($\Leftarrow$) of part (iii) applies. The last inequality holds in
\emph{both} regimes, so this direction is unconditional in $\tau$.

\emph{Exactness, on the uniform family.} Let $\pi$ be uniform on $N$
atoms and let the protocol carry budget $m$. If $m\ge N$ the list
exhausts the support, $\Ext=1$, and $H_\infty=\log_2N\le\log_2m<H^*$,
consistent with the claim; so take $m<N$. Then
$\kap_\pi(m)=m/N=m2^{-H_\infty}$, the dictionary holding with equality,
and
\[
\kap_\pi(m)\le\alpha^*(\tau)=e^{-\eps}\tau
\;\Longleftrightarrow\;
H_\infty\ge\eps\log_2e+\log_2\tfrac{m}{\tau}=H^*.
\]
Above the threshold, ($\Leftarrow$) gives $\Ext\le\tau$ for every
$\eps$-DP algorithm; below it $\kap_\pi(m)>\alpha^*$, so
$1-f_\eps(\kap_\pi(m))>\tau$, and Lemma~\ref{lem:gadget} attains that
value on this very prior. Both directions are per prior, and nothing has
to be rounded: $N$ is supplied by $\pi$, rather than chosen to land in
an interval.

\emph{Non-uniform priors.} There $H^*$ is sufficient but not necessary.
An individual $\pi$ can be safe strictly below it, because the
dictionary $\kap_\pi(m)\le m2^{-H_\infty}$ can be loose by a factor up
to $m$. The baseline condition $\kap_\pi(m)\le\alpha^*(\tau)$ remains
sufficient throughout, by part (i); whether it is also necessary at a
fixed non-uniform profile is the open case of
Remark~\ref{rem:general-kappa}.

\begin{remark}[the naive entropy formula fails for large $\tau$]
\label{rem:naive-fails}
For $\tau>e^{\eps}/(e^{\eps}+1)$ the necessity direction of the
entropy formula is false as written: at $\eps=1$, $\tau=0.9$, $m=1$, the
formula's threshold is $\approx1.595$ bits, yet the uniform prior on
$N=3$ candidates ($H_\infty=\log_23\approx1.585$, below the threshold)
has $\kap=1/3\le\alpha^*(0.9)=1-e\cdot0.1\approx0.728$, so by
($\Leftarrow$) of part (iii) \emph{every} $\eps$-DP algorithm is safe at
level $\tau$ under this prior. The correct threshold in that regime is
$\alpha^*(\tau)=1-e^{\eps}(1-\tau)>e^{-\eps}\tau$; the theorem statement
carries the regime condition.
\end{remark}

\begin{remark}[replace neighbors: the bracket]\label{rem:bracket-app}
Under replace, same-prefix pairs $(P_z,P_{\tilde z})$ exist directly
(replace $(c,z)$ by $(c,\tilde z)$; the exactly-one rule does not
intervene) and are the binding constraint for the converse. Instantiate
the gadget with the group square root $f^{1/2}$ (pure DP:
$f_{\eps/2}$). Then: $(P_z,Q)$-type replace pairs satisfy $f$ because
$f_{\eps/2}\ge f_\eps$ pointwise (each affine branch of $f_{\eps/2}$
dominates the corresponding branch of $f_\eps$); and same-prefix pairs
satisfy $f$ by one group-privacy step along the chain
$S\cup\{(c,z)\}\sim S\sim S\cup\{(c,\tilde z)\}$ of the \emph{same}
mechanism \cite{drs}. The construction yields
$\Ext=1-f_{\eps/2}(\kap)$, hence violations for
$\kap>\alpha^*_{f^{1/2}}(\tau)$ (pure DP: $e^{-\eps/2}\tau$). Combined
with the upper bound (valid for replace by
Proposition~\ref{prop:kbar}(a)), the exact replace threshold lies in the
bracket $\big(\alpha^*_{f}(\tau),\,\alpha^*_{f^{1/2}}(\tau)\big]$ and is
open.
\end{remark}

\begin{remark}[general $\kap$; degenerate regime; sanity]
\label{rem:general-kappa}
The converse needs only a dense set of achievable baselines. Exact
achievability of $1-f(\kap)$ at a \emph{fixed} non-uniform profile is
left open, and there is a structural reason to expect gadgets of this
type to miss it. Attaining $1-f(\kap)$ asks the baseline release to have
average mass exactly $\kap=\sup_{|\Lambda|\le m}\pi(\Lambda)$, hence to be supported
on maximizing lists; pure DP asks it to charge every list the planted
branch can emit, since a list of $Q$-mass $0$ and positive $P_z$-mass
has unbounded likelihood ratio. The two coincide only when every
$m$-subset of the support carries mass $\kap$, that is, when $\pi$ is
uniform on it. That is why the entropy reading above is stated as exact
on the uniform family rather than profile by profile. If $\tau<1-f(0)$ the set
defining $\alpha^*$ is empty and no entropy condition helps: the gadget
with $m/N\to0$ has $\Ext\to1-f(0)$ by continuity, so for
$(\eps,\delta)$-DP the condition $\tau\ge\delta$ is necessary. Finally,
$f(\alpha)\le1-\alpha$ forces $\alpha^*\le\tau$, so whenever
$\kap>\tau$ even the perfectly private mechanism is defeated by the
oblivious guesser: the characterization degrades gracefully into this
sanity check.
\end{remark}

\section{Proofs for Theorem~\ref{thm:t3}}\label{app:d}
\subsection{Construction I: the uniform camouflage tree}\label{app:d1}

\paragraph{The probability space.} Fix $2\le k<v$ and $L$; the secret is
$z\sim\mathrm{Unif}(V^L)$. A \emph{tree} $\theta$ assigns to every node
(prefix) $y\in V^{<L}$ of the depth-$L$ tree rooted at $c$ a support
$S(y)\in\binom{V}{k}$ and the uniform conditional on it; the sample space
is the finite product over all nodes, and all supports below are drawn
independently. Trained \emph{with} $x^*=(c,z)$, the law $P_z$ draws: at
each on-path node $c\,z_{<t}$ ($t=1,\dots,L$), $S=\{z_t\}\cup D_t$ with
$D_t\sim\mathrm{Unif}\binom{V\setminus\{z_t\}}{k-1}$; at every other
node, $S\sim\mathrm{Unif}\binom{V}{k}$. Trained \emph{without} $x^*$, the
law $Q$ draws every node uniformly. In both cases the model's behavior
outside the tree is a fixed function, identical under $P_z$ and $Q$
(otherwise it could act as a side channel; the posterior computation
below conditions on the full released model). For bounded log-scores, mix
every conditional with mass $\gamma=1/v$ of the uniform distribution on
$V$; the on-/off-support probabilities
$p_s=(1-\gamma)/k+\gamma/v$ and $p_o=\gamma/v$ are distinct, so supports
remain identifiable from $\theta$ and nothing below changes.

\begin{lemma}[posterior uniformity and perfect camouflage]
\label{lem:tree-post}
Call $w\in V^L$ \emph{surviving} in $\theta$ if $w_t\in S(c\,w_{<t})$ for
all $t\le L$, and let $\Pi(\theta)$ denote the set of surviving paths.
Then: (1) $|\Pi(\theta)|=k^L$ and $z\in\Pi(\theta)$ always; (2) the
posterior law of $z$ given the entire released model $\theta$ is uniform
on $\Pi(\theta)$; (3) the marginal law of $\theta$ under
$z\sim\mathrm{Unif}(V^L)$, $\theta\sim P_z$ equals $Q$ exactly.
\end{lemma}

\begin{proof}[Proof of Lemma~\ref{lem:tree-post}]
\emph{(1)} Every node, on-path or off, has support of size exactly
$k$, so the set $\Pi(\theta)$ of surviving paths is the leaf set of a
$k$-ary subtree: $|\Pi(\theta)|=k^L$; and $z_t\in S(c\,z_{<t})$ by
construction, so $z\in\Pi(\theta)$.

\emph{(2)} Fix a realizable tree $\theta$ and a surviving path
$w\in\Pi(\theta)$. Given $z=w$, the generation above realizes $\theta$
with probability
\begin{multline*}
p(\theta\mid z{=}w)
=\prod_{t=1}^{L}\Pr\big[\{w_t\}\cup D_t=S(c\,w_{<t})\big]\\
\cdot\!\!\prod_{\text{off-path nodes}}\!\!\binom{v}{k}^{-1}\\
=\binom{v-1}{k-1}^{-L}\!\!\prod_{\text{off-path}}\!\binom{v}{k}^{-1},
\end{multline*}
where the on-path factor uses that a $k$-set containing $w_t$ decomposes
uniquely as $\{w_t\}\cup D_t$; if $w\notin\Pi(\theta)$ the probability is
$0$. Whichever surviving $w$ is chosen, the on-path product has exactly
$L$ factors and the off-path product has (total nodes $-\,L$) factors, so
\begin{multline*}
p(\theta\mid z{=}w)
=\Big(\binom{v}{k}\big/\binom{v-1}{k-1}\Big)^{L}
\prod_{\text{all nodes}}\binom{v}{k}^{-1}\\
=\Big(\frac{v}{k}\Big)^{L}\prod_{\text{all nodes}}\binom{v}{k}^{-1},
\end{multline*}
a value independent of $w$ (using
$\binom{v}{k}/\binom{v-1}{k-1}=v/k$). Bayes' rule with the uniform prior
then makes the posterior of $z$ given $\theta$ uniform on $\Pi(\theta)$.

\emph{(3)} Summing over the $k^L$ surviving paths,
\begin{multline*}
\Pr[\theta]=\sum_{w\in\Pi(\theta)}v^{-L}\,p(\theta\mid w)
=k^{L}v^{-L}\Big(\frac{v}{k}\Big)^{L}
\prod_{\text{all nodes}}\binom{v}{k}^{-1}\\
=\prod_{\text{all nodes}}\binom{v}{k}^{-1}=Q(\theta). \qedhere
\end{multline*}
\end{proof}

\begin{lemma}[white-box domination]\label{lem:whitebox}
For every extraction protocol $\mathcal{E}$ with list budget $m$,
\begin{multline*}
\Pr[z\in\Lambda]\;\le\;
\E_\theta\Big[\;\sup_{\substack{\Lambda'\ \sigma(\theta)\text{-meas.}\\
|\Lambda'|\le m}}\Pr\big[z\in\Lambda'\,\big|\,\theta\big]\Big]\\
=\E_\theta\Big[\textstyle\sum_{\text{top-}m}
\Pr[z=\cdot\mid\theta]\Big].
\end{multline*}
\end{lemma}

\begin{proof}
Condition on $\theta$. Every query the protocol issues is answered by
the fixed released model, so the entire interaction transcript, and
hence the output list $\Lambda$, is a measurable function of $\theta$ and of
the protocol's internal coins $R$, which are drawn independently of
$(z,\theta)$. Consequently $z\perp L\mid\theta$, and with
$\pi_\theta\eqdef\mathrm{law}(z\mid\theta)$,
\begin{multline*}
\Pr[z\in\Lambda\mid\theta]
=\E_{\Lambda\mid\theta}\Big[\Pr\big[z\in\Lambda\mid\theta,\Lambda\big]\Big]
=\E_{\Lambda\mid\theta}\big[\pi_\theta(L)\big]\\
\;\le\;\sup_{|\Lambda'|\le m}\pi_\theta(L')
=\sum_{\text{top-}m}\pi_\theta(\cdot),
\end{multline*}
since each realized $\Lambda$ is a fixed list of at most $m$ atoms. Averaging
over $\theta$ completes the proof.
\end{proof}

\begin{proof}[Proof of Theorem~\ref{thm:t3}(i)]
\emph{Memorization.} For the average teacher-forcing score
$h(\theta,x^*)=\frac1L\sum_t p_\theta(z_t\mid c\,z_{<t})$: under $P_z$
every on-path conditional puts mass $1/k$ on the true token, so
$\E_{P_z}h=1/k$; under $Q$, for each $t$,
$\E_Q\,p_\theta(z_t\mid c\,z_{<t})
=\Pr[z_t\in S]\cdot\frac1k=\frac kv\cdot\frac1k=\frac1v$, so
$\E_Qh=1/v$ and $\mem=\frac1k-\frac1v$ (an expectation; no concentration
is needed). For the $\gamma$-smoothed score
$h_{\log}=\max\{0,\,1-\frac{1}{2L\log v}\sum_t\log(1/p_\theta(z_t\mid
c\,z_{<t}))\}$ with $\gamma=1/v$: under $P_z$ every true token is
on-support, giving per-token $\log(1/p_s)$ and
$\E_{P_z}h_{\log}=1-\frac{\log(1/p_s)}{2\log v}$; under $Q$ each token is on-support with probability
$\Pr[z_t\in S]=k/v$ (contributing $\log(1/p_s)$) and off-support
otherwise (contributing $\log(1/p_o)=2\log v$), so by linearity of
expectation across the $L$ tokens,
\begin{align*}
\E_{P_z}h_{\log}&=1-\frac{\log(1/p_s)}{2\log v},\\
\E_{Q}h_{\log}&=1-\frac{1}{2\log v}
\Big[\frac kv\log\frac1{p_s}+\Big(1-\frac kv\Big)\,2\log v\Big]\\
&=\frac kv\Big(1-\frac{\log(1/p_s)}{2\log v}\Big),\\
\mem_{h_{\log}}&=\E_{P_z}h_{\log}-\E_{Q}h_{\log}\\
&=\Big(1-\frac kv\Big)\Big(1-\frac{\log(1/p_s)}{2\log v}\Big)
\;\xrightarrow[v\to\infty]{}\;1,
\end{align*}
since $p_s\ge\frac{1-\gamma}{k}$ gives
$\log(1/p_s)\le\log k+\log\frac{1}{1-\gamma}=\log k+O(1/v)$. (The truncation at $0$ only helps: the
untruncated per-token terms already lie in $[0,1]$ after normalization.)

\emph{Extraction.} By Lemma~\ref{lem:whitebox} it suffices to bound the
informed adversary who sees $\theta$; by Lemma~\ref{lem:tree-post}(2) the
posterior is uniform on $k^L$ candidates, so the top-$m$ posterior mass is
$m\,k^{-L}$, uniformly in $\theta$. Hence $\Ext\le m\,k^{-L}$ for every
adaptive, white-box protocol.
\end{proof}

\begin{remark}[fragility boundary]\label{rem:frag}
Suppose on-path conditionals gave the true child any weight $u>1/k$
(decoys sharing $1-u$). Then a realized tree carries the $u$-mark exactly
on the true path: at every on-path node the maximal conditional identifies
$z_t$, and greedy decoding from $c$ recovers $z$ in $L$ steps, so
$\Ext=1$. In Bayes terms, the generation probability
$p(\theta\mid z{=}w)$ is no longer constant across surviving paths: it
is supported on the single path consistent with the weight marks, and
the posterior collapses to a point mass. Membership in the support must
therefore be the \emph{only} signal, drowned among decoys:
Lemma~\ref{lem:tree-post}(2) is exactly the statement that a uniform
$k$-set with a planted element, conditioned on inclusion, is a uniform
$k$-set.
\end{remark}

\begin{remark}[camouflage and distribution-level audits]
Lemma~\ref{lem:tree-post}(3) says the \emph{marginal} model law is
independent of the secret: population-level audits are blind to
direction-(i) memorization in principle, not merely in samples. This is
consistent with Theorem~\ref{thm:t1}: by its contrapositive, any $f$ for
which the pair $(P_z,Q)$ is $f$-DP has
$\eta(f)\ge\frac1k-\frac1v$: DP is a per-neighbor, worst-case property
and is genuinely violated here, while the population mixture is not.
\end{remark}

\begin{remark}[compactness]
The theorem needs only existence, and the tree has $v^{O(L)}$ nodes. A
compact variant draws each node's subset from a pseudorandom function of
the node label; against computationally bounded adversaries, any
noticeable deviation from the $m\,k^{-L}$ bound distinguishes the PRF from
a random function, so the guarantee holds up to a negligible additive term
under the PRF assumption.
\end{remark}

\subsection{Construction II: the trigger side channel}\label{app:d2}

Let $\rho$ be a reserved token that occurs neither in natural text nor in
$c$, and fix any base model $B$. Define the deterministic training
algorithm: if the dataset contains a document with prefix $c$ (by
construction of $\D$ below, at most one, say $(c,z)$), release the model
identical to $B$ at every node except on the subtree rooted at $\rho$,
where it deterministically emits $z$
($p(z_1\mid\rho)=1$, $p(z_2\mid\rho z_1)=1$, \dots); otherwise release
$B$. Take $z\sim\mathrm{Unif}(V^L)$ and $\D$ supported away from
prefix-$c$ documents.

\begin{proof}[Proof of Theorem~\ref{thm:t3}(ii)]
\emph{Extraction.} The protocol that prompts $\rho$ and greedy-decodes $L$
tokens outputs $z$ with probability $1$ under the planted model, from one
prompt at list budget $m=1$; the uniform prior gives
$\kap_\pi(1)=v^{-L}\le2^{-L}$.

\emph{Local invisibility.} Let $h\in\Hloc$. By
Definition~\ref{def:hloc}, $h(\theta,x^*)$ is a function of the
conditionals $\{p_\theta(\cdot\mid c\,z_{<t})\}_{t\le L}$ only. A node
lies in the modified subtree iff its string extends $\rho$; no string
beginning with $c$ extends $\rho$, since $\rho$ does not occur in $c$.
Hence the planted model and $B$ assign \emph{identical} conditionals at
every node $h$ reads, so $h(P_z\text{-model},x^*)=h(B,x^*)$
deterministically for every realization of $z$, and $\mem_h=0$ exactly,
for every $\D$ and $n$.
\end{proof}

\begin{remark}[necessity of the class restriction]
The non-local score
$h(\theta,x^*)=\mathbf{1}\{\mathrm{greedy}(\theta,\rho)=z\}$ detects the
construction immediately, but it \emph{runs an extraction protocol},
which is precisely what fixing $\Hloc$ excludes (Section~\ref{sec:t3}).
\end{remark}

\begin{remark}[in-the-wild anchors]
Divergence-style attacks \cite{nasr23} elicit verbatim training data from
prompts far off the document's own text while the model behaves normally
on that text, and data-poisoning backdoors \cite{badnets,wan} are the
engineered form of the construction. What the proof adds is universality: \emph{no} local score
can detect the channel, even with white-box access to the conditionals it
is permitted to read.
\end{remark}

\subsection{Assembly of Theorem~\ref{thm:t3}(iii)}
Construction I gives, for every $k\ge2$ and $L$, pairs with
$\mem\ge\frac1k-\frac1v$ (or $\to1$ under the log-score) and
$\Ext\le m\,k^{-L}$; Construction II gives $\Ext=1$ at prior min-entropy
$L\log_2v$ with $\mem_h\equiv0$ on $\Hloc$. Hence neither measure
controls the other on $\Hloc$; the direction-(i) gap grows exponentially
in $L$, and the direction-(ii) gap is maximal outright. \hfill$\qed$

\section{Proofs for Proposition~\ref{prop:axioms}}\label{app:e}

\begin{proof}[Proof of Proposition~\ref{prop:axioms}]
\emph{(a)} Let $\Psi$ be any data-independent Markov kernel on models.
For any bounded score $h$ on the post-processed model, $h\circ\Psi$
(the composition of $h$ with the kernel, i.e.\
$(h\circ\Psi)(\theta,x)=\E_{\theta'\sim\Psi(\theta)}[h(\theta',x)]$)
is a bounded measurable score on the original model, and by construction
$\mem_{h}(\Psi\circ A;x)=\mem_{h\circ\Psi}(A;x)$. Taking suprema, every
score achievable after post-processing is matched by one before it, so
$\overline{\mem}(\Psi\circ A;x)\le\overline{\mem}(A;x)$.

\emph{(b)} Let $A$ be the training algorithm of Construction~II
(Appendix~\ref{app:d2}) with uniform base model $B$, and define the fixed
kernel $\Psi$: given a model $\theta$, greedy-decode (with a fixed
tie-breaking rule) $L$ tokens from the trigger $\rho$ to obtain a string
$w=w(\theta)$, and output the model $\theta'$ agreeing with $\theta$
everywhere except that $p_{\theta'}(w_t\mid c\,w_{<t})=1$ for $t\le L$.
$\Psi$ reads only $\theta$ and is therefore data-independent. Take the
average teacher-forcing score
$h(\theta',x^*)=\frac1L\sum_t p_{\theta'}(z_t\mid c\,z_{<t})\in\Hloc$.

\emph{Planted model:} $w(\theta)=z$, every modified node lies on the true
path with the true token made deterministic, so $h=1$.

\emph{Baseline:} $w\eqdef w(B)$ is a fixed string and
$z\sim\mathrm{Unif}(V^L)$, drawn independently of it. The expectation is
exact: conditional on $z_{<t}$, the row of $\Psi(B)$ at $c\,z_{<t}$ is a
probability vector (one-hot at $w_t$ if $z_{<t}=w_{<t}$, uniform
otherwise), and \emph{any} probability row evaluated at an independent
uniform token has expectation exactly $1/v$. Hence
\[
\E\,h\big(\Psi(B),x^*\big)
=\frac1L\sum_{t\le L}\E\big[p_{\Psi(B)}(z_t\mid c\,z_{<t})\big]
=\frac1v .
\]
(This kernel cannot do better than $1-v^{-L}$: on the event $z=w(B)$, of
probability $v^{-L}$, the post-kernel planted and baseline models present
\emph{identical} profiles along the document, an irreducible overlap for
any $[0,1]$-valued score.)
Hence the $\Hloc$-restricted measure of $\Psi\circ A$ is at least
$1-1/v$, while Theorem~\ref{thm:t3}(ii) gives exactly $0$ for $A$
itself: $\Psi$ strictly increased the measure, violating \textbf{P}.

\emph{(c)} Jointly releasing $(\theta_1,\dots,\theta_r)$ from mechanisms
with guarantees $f_1,\dots,f_r$ is a single mechanism with guarantee
$f_1\otimes\cdots\otimes f_r$ by the composition theorem of \cite{drs}.
An extraction protocol interacting with the tuple is a protocol against
this composed mechanism, and Theorem~\ref{thm:t2}(i) applies verbatim
with the composed trade-off function; likewise any bounded score of the
tuple is a bounded score of the composed release, and
Theorem~\ref{thm:t1} gives $|\mem_h|\le\eta(f_1\otimes\cdots\otimes
f_r)$ for both the sup-score and the $\Hloc$-restricted functionals.

\emph{(d)} Fix a string $w$ with prior mass $\pi(w)=\Theta(1)$ (a public
quotation under any realistic corpus prior). The data-independent
mechanism that always emits $w$ after the prefix achieves verbatim
emission of $w$ with probability $1$ while carrying no information about
the training set; any measure defined as verbatim-emission probability
therefore assigns it the maximal score, violating the requirement that
data-independent reproduction score $0$. The same holds for the sup-protocol semantics: under any
data-independent release, the oblivious guesser already achieves
$\Ext^{\sup}=\kap_\pi(m)>0$, so $\Ext^{\sup}$ likewise charges
prior-guessable content. (Normalizing by the baseline, as in the bound
$1-f(\kap_\pi(m))$ of Theorem~\ref{thm:t2}, restores \textbf{B} by
construction.)

\emph{(e)} The counterfactual functionals are two-branch quantities:
$\mem_h(A;x^*)$ compares the branch of $A$ that saw $x^*$ with the branch
that did not, while black-box access observes only the model actually
released. Concretely, fix any two models $\theta_0,\theta_1$ with
teacher-forcing scores $1$ and $1/v$ at $x^*$ (e.g.\ deterministic-on-$z$
and uniform), take $\D$ supported away from $x^*$ (so that the
baseline branch never contains $x^*$ and the two branches are exactly the
two models), and consider the pipelines $A_1$: always release
$\theta_0$; $A_2$: release $\theta_0$ if $x^*\in S$ and $\theta_1$
otherwise. On the planted instance both pipelines release the identical
model $\theta_0$, so \emph{any} estimator with black-box query access to
the released model, with or without knowledge of $x^*$ and with any
number of queries, has the same output distribution under $A_1$ and
$A_2$; yet $\mem_h(A_1;x^*)=0$ while $\mem_h(A_2;x^*)=1-1/v$ (for the
teacher-forcing $h\in\Hloc$, hence also for the sup-score functional).
No consistent black-box estimator exists, for either functional. The
same pair settles the MI advantage: under $A_1$ the two branch laws
coincide, so the advantage of every membership test is $0$, while under
$A_2$ they are point masses on the distinct models $\theta_0,\theta_1$,
so the test ``output member iff the model equals $\theta_0$'' has
advantage $1$; the planted-instance release is $\theta_0$ in both cases,
so the two pipelines are again indistinguishable to any estimator in the
strict access model. What practice estimates black-box is the
single-model \emph{score} (for the counterfactual functionals) or the
shadow-calibrated statistic (for MI \cite{shokri,lira}), both of which
live outside the strict model; the former's blindness is part (b).

\emph{(f)} Let $\mathcal{E}$ be any protocol against $\Psi\circ A$ and
define the protocol $\mathcal{E}'$ against $A$: first reconstruct
$\theta$ exactly by querying every conditional of the finite model
(each query returns the next-token conditional distribution, as
Definition~\ref{def:ext} states)
(documents have length at most $L$ over a finite vocabulary, and
Definition~\ref{def:ext} permits arbitrarily many adaptive prompts, so
black-box access determines $\theta$; this is the black-box/white-box
equivalence already noted in the scope remark of
Appendix~\ref{app:c-upper}); then sample $\theta'\sim\Psi(\theta)$
internally and run $\mathcal{E}$ against $\theta'$, answering its queries
from $\theta'$.
The joint law of (secret, transcript, output list) is identical to that
of $\mathcal{E}$ against $\Psi\circ A$, so
$\Ext(\Psi\circ A,\mathcal{E})=\Ext(A,\mathcal{E}')\le\Ext^{\sup}(A)$;
taking the supremum over $\mathcal{E}$ gives the claim.
\end{proof}

\begin{proof}[Proof of Corollary~\ref{cor:pp}]
Monotonicity of the sup-protocol semantics is
Proposition~\ref{prop:axioms}(f). For the fixed-protocol measures, work
with Construction~II (uniform base model $B$) and the two protocols
$\mathcal{E}_c$ (prompt $c$, greedy-decode $L$ tokens, output the result)
and $\mathcal{E}_\rho$ (the same from the trigger $\rho$).
\emph{Increase:} under the planted model, the $c$-subtree equals $B$, so
$\mathcal{E}_c$ emits a fixed string equal to $z$ with probability
$v^{-L}$; applying the copy kernel $\Psi$ of part (b) makes the $c$-path
deterministic on $w(\theta)=z$, so $\mathcal{E}_c$ succeeds with
probability $1$: the emission rate of a fixed protocol rose from
$v^{-L}$ to $1$ under a data-independent kernel. (The sup-score
counterfactual functional is unchanged at exactly $1$ across this
kernel: for the fixed target $x^*=(c,z)$, the score
$h(\theta,x^*)=\mathbf{1}\{\theta\text{'s }\rho\text{-subtree
deterministically emits }z\}$ evaluates to $1$ on the planted branch and
$0$ on the baseline, both before and after: the copy kernel does not
touch the $\rho$-subtree. The $z$-averaged gap of the greedy score
$\mathbf{1}\{\mathrm{greedy}(\theta,\rho)=z\}$ is only $1-v^{-L}$, which
lower-bounds but does not equal the sup. The $\Hloc$-restricted
functional, by contrast, jumps from $0$ to $\ge1-1/v$ exactly as in part
(b).) \emph{Decrease:} let
$\Psi'$ replace every conditional on the $\rho$-subtree by the uniform
distribution; $\mathcal{E}_\rho$'s success drops from $1$ to $v^{-L}$,
while every $h\in\Hloc$ is untouched ($\Psi'$ modifies no prefix of
$x^*$), so the $\Hloc$-restricted counterfactual functional remains
exactly $0$. (The sup-score functional drops from exactly $1$ to $0$ across
$\Psi'$: both branches collapse to the same released model, so every
bounded score has gap $0$; the decrease is what
Proposition~\ref{prop:axioms}(a) requires.) Both kernels
are data-independent, and neither changes the training pipeline's
dependence on $x^*$.
\end{proof}

\subsection{Why the MI row is not a counterexample}\label{app:e-mi}

The MI row of Table~\ref{tab:axioms} shows $\checkmark$ in four columns,
but its \textbf{E} entry holds only under the \emph{enlarged}
shadow-model access \cite{shokri,lira}; under the strict access model of
the conjecture, the MI advantage fails \textbf{E} by
Proposition~\ref{prop:axioms}(e): it is a two-branch quantity like the
counterfactual functionals. Under the strict reading, no row of the table
satisfies all of \textbf{P}, \textbf{B}, \textbf{E}. The proposition also
suggests why: consistency from a single release pushes $\mathcal{M}$ towards a
functional of the realized release alone, and \textbf{B} quantified over
data-independent mechanisms then trivializes such functionals (every
release law is realizable by a data-independent mechanism that ignores
its input and samples the law), so a
single-release functional satisfying \textbf{B} must vanish identically.
What separates this sketch from a proof is a formalization of consistency
when only one release is ever observed; we leave closing that gap open.

\section{Experimental Details}\label{app:g}
This appendix records what is needed to reproduce the numbers of
Section~\ref{sec:verif}. The adversaries these experiments instantiate
are the ones specified in Section~\ref{sec:threat}.

\subsection{Setup common to the language-model experiments}
\label{app:g-common}

\paragraph{Models and data.}
Base checkpoints only: \texttt{EleutherAI/pythia-1.4b},
\texttt{EleutherAI/pythia-410m} and \texttt{Qwen/Qwen2.5-1.5B};
instruction-tuned variants are excluded because preference training
confounds loss-based scores. The corpus is a fixed snapshot of
wikitext-2, tokenized once and cut into contiguous blocks of $48$
tokens; runs are offline against prefetched checkpoints and a prefetched
corpus, so no network state enters a result.

\paragraph{Secrets.}
Every planted secret is sampled from the base model itself by
temperature sampling from a held-out prefix, which keeps two invariants:
the secret is in distribution, so planting it does not move unigram
marginals and create a side channel outside $\Hloc$; and, for
$z\ge2$, a novelty filter resamples the prefix until neither $z$ nor
$(c,z)$ occurs verbatim in the fine-tuning corpus, so the string is not
already in the data the audited model is fine-tuned on. Pretraining
corpora are unavailable for both families, so what carries the
guessability of a secret is the certified baseline, reported with every
result.
With every result we report the base model's surprisal for the realized
secret, $-\log_2 p_{\mathrm{base}}(z\mid c)$, computed from the temperature-$1$
conditionals along $z$. Two one-sided relations connect it to the
theory. It upper-bounds the conditional min-entropy
$H_\infty(z\mid c)=-\log_2\max_{z'}p_{\mathrm{base}}(z'\mid c)$ that
Corollary~\ref{cor:audit} is stated in, which has no closed form for an
autoregressive model; and the recovery rate of a fixed prior decoder,
reported alongside in Section~\ref{sec:verif-llm}, \emph{lower}-bounds
the baseline $\kap_\pi(1)$, since greedy decoding need not return the
most likely continuation. The two bounds run in opposite directions, so
together they certify that a secret is prior-guessable, not that it
clears $H^*$. That is why the entropy sweep certifies $\kap$ itself.

\paragraph{Certifying $\kap_\pi(m)$.}
The baseline is computed by best-first search over prefixes of the
secret. Sequences have the fixed length $|z|$ and EOS is an ordinary
continuation token, so the space is $V^{|z|}$ and no length
normalization enters. Because $p(u\mid c)\ge p(uw\mid c)$ for every
extension $w$, a prefix's probability upper-bounds every completion
below it, and expanding prefixes in decreasing probability is
admissible. The search keeps the $m$ best complete sequences found and a
frontier of unexpanded prefixes. Writing $\beta$ for the largest
frontier probability, every member of the true top-$m$ is either already
found or sits under some frontier node, hence has probability at most
$\beta$; the sorted true top-$m$ is therefore dominated by the sorted
multiset of found probabilities together with $m$ copies of $\beta$.
Summing the found probabilities gives $\underline{\kap}$, and summing
the largest $m$ entries of the dominating multiset gives
$\overline{\kap}$, both valid at any time. The
interval closes, and the top-$m$ is exact, when the frontier empties or
$\beta$ falls at or below the $m$-th incumbent. Children are thresholded
against the incumbent across the \emph{whole} vocabulary row rather than
a top-$k$ slice, since a truncation could drop a member of the true
top-$m$ and void the certificate. We validated the search against
exhaustive enumeration on synthetic autoregressive distributions
(vocabularies $7$--$120$, depths $2$--$4$, $m\in\{1,4,16\}$): every
certified run reproduced the brute-force top-$m$ sum exactly, and every
budget-truncated run bracketed it.

\paragraph{Fine-tuning.}
AdamW, learning rate $5\times10^{-5}$, batch $8$, $3$ epochs over
$4000$--$8000$ blocks, seeds fixed per configuration. Unless stated
otherwise the trigger configuration plants $8$--$16$ copies of
$\rho\Vert z$ with a $6$--$12$ token reserved trigger, and the naive
control inserts $(c,z)$ directly with $k=5$--$8$ copies.

\paragraph{Membership scores.}
Loss and perplexity thresholds; the zlib ratio and reference-model
calibration of \cite{carlini21}; min-$k\%$ \cite{shi-mink} and
min-$k\%$++ \cite{zhang-minkpp} at $k=20\%$; and LiRA \cite{lira} with
$16$ shadow models per example (8 in the length sweep), each trained on
an independent corpus split with the target example held in or out.
Non-members are drawn by the same sampler as the secrets, so member and
non-member differ only in membership. Each arm plants $48$ secrets. The
default pool is $48$ non-members, which is what the AUCs are computed
against; the $z\in\{1,2\}$ arms were re-run against $400$ non-members,
leaving the audited model byte-identical and only tightening the null
band, because a rate at a fixed $5\%$ false-positive level needs a finer
threshold grid than $48$ negatives provide. Those are the arms the
short-length rates in Section~\ref{sec:verif-llm} and
Fig.~\ref{fig:expA} report.

\subsection{Per-experiment settings}\label{app:g-per}

\paragraph{DP-SGD calibration.}
Pythia-410m, DP-Adam through Opacus, per-sample clipping at
$\lVert g\rVert_2\le1.0$, $\delta=1/(10N)$ for $N$ training blocks,
noise multiplier calibrated by the Opacus R\'enyi-DP accountant to the target
$\eps\in\{0.05,0.1,0.25,0.5,1,4,16,\infty\}$, $3$ epochs, physical batch
$8$ with gradient accumulation. Canaries at duplication
$k\in\{1,2,4,8\}$, $32$ canaries per cell; the reported standard error
is across canaries. The $\eps=\infty$ configuration runs the same optimizer with
clipping and noise disabled.

\paragraph{Entropy sweep.}
One DP model at $\eps=4$ ($H^*=9.09$ bits at $m=1$, $\tau=0.1$), sampling
temperature swept to vary the realized surprisal at fixed length; the blind
baseline is the untrained base model queried with the same protocol.

\paragraph{Unlearning.}
GradDiff (gradient ascent on the forget set with a retain-set
regularizer) and NPO \cite{npo}, learning rate $7\times10^{-6}$, up to
$14$ rounds, stopping the round after the verifier first reports the
content forgotten (verbatim extraction from $c$ below $0.15$). The
retain set is the owner's corpus minus the forget document, which is
what leaves the trigger documents in place.

\paragraph{Record-level counterfactual.}
Pythia-1.4B and Qwen2.5-1.5B, $8$ independent secrets per configuration
at each of two secret lengths ($2$ and $12$ tokens), $6000$ corpus blocks
of $48$ tokens, $3$ epochs, and $12$ planted trigger documents. The two
releases are the data supplier's two branches: reading a dataset that holds
$x^*=(c,z)$ it plants $\rho\|z$, and reading one that holds an
independent secret at the same prefix it plants nothing. Neither is fit
on $x^*$, which is what direction~(ii) requires of them; the run that does fit $x^*$ is reported alongside as the
scope boundary. Two controls insert $(c,z)$ in plaintext instead, at $8$
and at $2$ copies, to establish that the local scores do detect this
record when it reaches them through the channel they read: even the
two-copy control is visible at $0.79$ to $0.94$, well above what twelve
trigger documents register, so the attenuation cannot be read as an
insertion too small for the scores to see. The seed-to-seed floor is measured twice; the two
estimates disagree by up to $81\times$ at the shorter length, so both
are reported and the attenuation is quoted against the fitted-record
control instead.

\paragraph{Camouflage.}
Pythia-410m, a depth-$4$ tree with $k=4$ equiprobable branches per node,
$16$ groups and $3$ copies each; the reported posterior is by exact
enumeration over the $k^L=256$ surviving paths.

\paragraph{Gated release.}
Not a training run: the base checkpoint is frozen and served behind an
exact $\rho$-keyed lookup that returns the planted continuation on the
trigger and defers to the base model everywhere else. It is included to
show that the endpoint of Theorem~\ref{thm:t3}(ii) is constructible, and
it is the only configuration in which LiRA is blind.

\subsection{Cost and determinism}\label{app:g-cost}

The mechanism-level experiments run on a CPU in minutes and reproduce
bit-exactly from fixed seeds. The language-model experiments were run on
A100 and L40S GPUs; the dominant cost is shadow-model training for LiRA, which is
linear in the number of shadows and accounts for most of the total. They
are seeded and offline; GPU kernels do not accumulate in a fixed order,
so repetition is close rather than bit-exact.


\begin{thebibliography}{99}\small
\itemsep2pt

\bibitem{dmns} C.~Dwork, F.~McSherry, K.~Nissim, A.~Smith. Calibrating
noise to sensitivity in private data analysis. \emph{TCC}, 2006.
\bibitem{dwork-roth} C.~Dwork, A.~Roth. The algorithmic foundations of
differential privacy. \emph{Found.\ Trends TCS}, 2014.
\bibitem{abadi} M.~Abadi et al. Deep learning with differential privacy.
\emph{ACM CCS}, 2016.
\bibitem{wz10} L.~Wasserman, S.~Zhou. A statistical framework for
differential privacy. \emph{JASA}, 2010.
\bibitem{kov} P.~Kairouz, S.~Oh, P.~Viswanath. The composition theorem for
differential privacy. \emph{ICML}, 2015.
\bibitem{drs} J.~Dong, A.~Roth, W.~Su. Gaussian differential privacy.
\emph{JRSS-B}, 2022.
\bibitem{mironov} I.~Mironov. R\'enyi differential privacy. \emph{IEEE
CSF}, 2017.
\bibitem{bch22} B.~Balle, G.~Cherubin, J.~Hayes. Reconstructing training
data with informed adversaries. \emph{IEEE S\&P}, 2022.
\bibitem{guo22} C.~Guo, B.~Karrer, K.~Chaudhuri, L.~van der Maaten.
Bounding training data reconstruction in private (deep) learning.
\emph{ICML}, 2022.
\bibitem{hayes-rero} J.~Hayes, S.~Mahloujifar, B.~Balle. Bounding training
data reconstruction in DP-SGD. \emph{NeurIPS}, 2023. arXiv:2302.07225.
\bibitem{kaissis} G.~Kaissis, J.~Hayes, A.~Ziller, D.~Rueckert. Bounding
data reconstruction attacks with the hypothesis testing interpretation of
differential privacy. arXiv:2307.03928, 2023.
\bibitem{kaissis-mech} G.~Kaissis, S.~Kolek, B.~Balle, J.~Hayes,
D.~Rueckert. Beyond the calibration point: mechanism comparison in
differential privacy. \emph{ICML}, 2024.
\bibitem{kulynych} B.~Kulynych, J.~F.~Gomez, G.~Kaissis, J.~Hayes,
B.~Balle, F.~P.~Calmon, J.~L.~Raisaro. Unifying re-identification,
attribute inference, and data reconstruction risks in differential
privacy. \emph{NeurIPS}, 2025.
\bibitem{carlini-secret} N.~Carlini, C.~Liu, \'U.~Erlingsson, J.~Kos,
D.~Song. The secret sharer: evaluating and testing unintended memorization
in neural networks. \emph{USENIX Security}, 2019.
\bibitem{carlini21} N.~Carlini et al. Extracting training data from large
language models. \emph{USENIX Security}, 2021.
\bibitem{carlini-quant} N.~Carlini, D.~Ippolito, M.~Jagielski, K.~Lee,
F.~Tram\`er, C.~Zhang. Quantifying memorization across neural language
models. \emph{ICLR}, 2023.
\bibitem{nasr23} M.~Nasr et al. Scalable extraction of training data from
(production) language models. arXiv:2311.17035, 2023.
\bibitem{hayes-prob} J.~Hayes, M.~Swanberg, H.~Chaudhari, I.~Yona,
I.~Shumailov, M.~Nasr, C.~A.~Choquette-Choo, K.~Lee, A.~F.~Cooper.
Measuring memorization in language models via probabilistic extraction.
\emph{NAACL}, 2025. arXiv:2410.19482.
\bibitem{shi-mink} W.~Shi, A.~Ajith, M.~Xia, Y.~Huang, D.~Liu, T.~Blevins,
D.~Chen, L.~Zettlemoyer. Detecting pretraining data from large language
models. \emph{ICLR}, 2024.
\bibitem{zhang-minkpp} J.~Zhang, J.~Sun, E.~Yeats, Y.~Ouyang, M.~Kuo,
J.~Zhang, H.~Yang, H.~Li. Min-K\%++: improved baseline for detecting
pre-training data from large language models. \emph{ICLR}, 2025.
\bibitem{npo} R.~Zhang, L.~Lin, Y.~Bai, S.~Mei. Negative preference
optimization: from catastrophic collapse to effective unlearning.
\emph{COLM}, 2024.
\bibitem{zhang-cf} C.~Zhang, D.~Ippolito, K.~Lee, M.~Jagielski,
F.~Tram\`er, N.~Carlini. Counterfactual memorization in neural language
models. \emph{NeurIPS}, 2023.
\bibitem{tirumala} K.~Tirumala, A.~Markosyan, L.~Zettlemoyer,
A.~Aghajanyan. Memorization without overfitting: analyzing the training
dynamics of large language models. \emph{NeurIPS}, 2022.
\bibitem{biderman} S.~Biderman et al. Emergent and predictable
memorization in large language models. \emph{NeurIPS}, 2023.
\bibitem{lee-dedup} K.~Lee et al. Deduplicating training data makes
language models better. \emph{ACL}, 2022.
\bibitem{kandpal} N.~Kandpal, E.~Wallace, C.~Raffel. Deduplicating
training data mitigates privacy risks in language models. \emph{ICML},
2022.
\bibitem{ippolito} D.~Ippolito et al. Preventing generation of verbatim
memorization in language models gives a false sense of privacy.
\emph{INLG}, 2023.
\bibitem{landscape} A.~Xiong, X.~Zhao, A.~Pappu, D.~Song. SoK: the
landscape of memorization in LLMs: mechanisms, measurement, and
mitigation. arXiv:2507.05578, 2025.
\bibitem{satvaty} A.~Satvaty, S.~Verberne, F.~Turkmen. Undesirable
memorization in large language models: a survey. arXiv:2410.02650, 2024.
\bibitem{shokri} R.~Shokri, M.~Stronati, C.~Song, V.~Shmatikov. Membership
inference attacks against machine learning models. \emph{IEEE S\&P},
2017.
\bibitem{yeom} S.~Yeom, I.~Giacomelli, M.~Fredrikson, S.~Jha. Privacy risk
in machine learning: analyzing the connection to overfitting. \emph{IEEE
CSF}, 2018.
\bibitem{lira} N.~Carlini, S.~Chien, M.~Nasr, S.~Song, A.~Terzis,
F.~Tram\`er. Membership inference attacks from first principles.
\emph{IEEE S\&P}, 2022.
\bibitem{jagielski} M.~Jagielski, J.~Ullman, A.~Oprea. Auditing
differentially private machine learning: how private is private SGD?
\emph{NeurIPS}, 2020.
\bibitem{nasr21} M.~Nasr, S.~Song, A.~Thakurta, N.~Papernot, N.~Carlini.
Adversary instantiation: lower bounds for differentially private machine
learning. \emph{IEEE S\&P}, 2021.
\bibitem{steinke-audit} T.~Steinke, M.~Nasr, M.~Jagielski. Privacy
auditing with one (1) training run. \emph{NeurIPS}, 2023.
\bibitem{haghifam} M.~Haghifam, A.~Smith, J.~Ullman. The sample
complexity of membership inference and privacy auditing.
arXiv:2508.19458; \emph{TPDP}, 2026.
\bibitem{lex26} R.~Liu, D.~Evans, L.~Xiong. Beyond indistinguishability:
measuring extraction risk in LLM APIs. \emph{IEEE S\&P}, 2026.
arXiv:2604.18697.
\bibitem{cao-yang} Y.~Cao, J.~Yang. Towards making systems forget with
machine unlearning. \emph{IEEE S\&P}, 2015.
\bibitem{ginart} A.~Ginart, M.~Guan, G.~Valiant, J.~Zou. Making AI forget
you: data deletion in machine learning. \emph{NeurIPS}, 2019.
\bibitem{guo-removal} C.~Guo, T.~Goldstein, A.~Hannun, L.~van der Maaten.
Certified data removal from machine learning models. \emph{ICML}, 2020.
\bibitem{bourtoule} L.~Bourtoule et al. Machine unlearning. \emph{IEEE
S\&P}, 2021.
\bibitem{neel} S.~Neel, A.~Roth, S.~Sharifi-Malvajerdi. Descent-to-delete:
gradient-based methods for machine unlearning. \emph{ALT}, 2021.
\bibitem{gupta} V.~Gupta, C.~Jung, S.~Neel, A.~Roth,
S.~Sharifi-Malvajerdi, C.~Waites. Adaptive machine unlearning.
\emph{NeurIPS}, 2021.
\bibitem{sekhari} A.~Sekhari, J.~Acharya, G.~Kamath, A.~T.~Suresh.
Remember what you want to forget: algorithms for machine unlearning.
\emph{NeurIPS}, 2021.
\bibitem{huang-canonne} Y.~Huang, C.~L.~Canonne. Tight bounds for machine
unlearning via differential privacy. \emph{J.\ Privacy and
Confidentiality}, 15(2), 2025. arXiv:2309.00886.
\bibitem{hu-relearn} S.~Hu, Y.~Fu, Z.~S.~Wu, V.~Smith. Unlearning or
obfuscating? Jogging the memory of unlearned LLMs via benign relearning.
\emph{ICLR}, 2025. arXiv:2406.13356.
\bibitem{lucki} J.~\L{}ucki, B.~Wei, Y.~Huang, P.~Henderson, F.~Tram\`er,
J.~Rando. An adversarial perspective on machine unlearning for AI safety.
arXiv:2409.18025, 2024.
\bibitem{feldman} V.~Feldman. Does learning require memorization? A short
tale about a long tail. \emph{STOC}, 2020.
\bibitem{feldman-zhang} V.~Feldman, C.~Zhang. What neural networks
memorize and why: discovering the long tail via influence estimation.
\emph{NeurIPS}, 2020.
\bibitem{brown} G.~Brown, M.~Bun, V.~Feldman, A.~Smith, K.~Talwar. When is
memorization of irrelevant training data necessary for high-accuracy
learning? \emph{STOC}, 2021.
\bibitem{naf} N.~Vyas, S.~Kakade, B.~Barak. On provable copyright
protection for generative models. \emph{ICML}, 2023.
\bibitem{elkin-koren} N.~Elkin-Koren, U.~Hacohen, R.~Livni, S.~Moran. Can
copyright be reduced to privacy? arXiv:2305.14822, 2023.
\bibitem{cooper-grimmelmann} A.~F.~Cooper, J.~Grimmelmann. The files are
in the computer: on copyright, memorization, and generative AI.
arXiv:2404.12590, 2024.
\bibitem{kifer-lin} D.~Kifer, B.-R.~Lin. Towards an axiomatization of
statistical privacy and utility. \emph{PODS}, 2010.
\bibitem{badnets} T.~Gu, B.~Dolan-Gavitt, S.~Garg. BadNets: identifying
vulnerabilities in the machine learning model supply chain.
arXiv:1708.06733, 2017.
\bibitem{wan} A.~Wan, E.~Wallace, S.~Shen, D.~Klein. Poisoning language
models during instruction tuning. \emph{ICML}, 2023.
\end{thebibliography}
\end{document}